\documentclass[english,final]{IEEEtran}
\usepackage[T1]{fontenc}
\usepackage[latin9]{inputenc}
\usepackage{color}
\usepackage{float}
\usepackage{mathrsfs}
\usepackage{bm}
\usepackage{amsmath}
\usepackage{amsthm}
\usepackage{amssymb}
\usepackage{graphicx}

\makeatletter

\floatstyle{ruled}
\newfloat{algorithm}{tbp}{loa}
\providecommand{\algorithmname}{Algorithm}
\floatname{algorithm}{\protect\algorithmname}

\theoremstyle{plain}
\newtheorem{prop}{\protect\propositionname}
\theoremstyle{definition}
\newtheorem{defn}{\protect\definitionname}
\theoremstyle{definition}
\newtheorem{condition}{\protect\conditionname}
\theoremstyle{plain}
\newtheorem{thm}{\protect\theoremname}
\theoremstyle{plain}
\newtheorem{lem}{\protect\lemmaname}

\usepackage{amsmath,amssymb}
\usepackage{color}  
\usepackage{graphicx,psfrag,cite,subfigure}

\usepackage{cite}

\author{
Junting~Chen, \textit{Member,~IEEE}, and David~Gesbert, \textit{Fellow,~IEEE}
\thanks{This paper has been accepted in IEEE Transactions on Wireless Communications.
This research was funded, in part, by one or more of the following grants:
No. ZDSYS201707251409055, No. 2017ZT07X152, No. 2018B030338001, No. 2018YFB1800800, and  
ERC under the European Union Horizon 2020 research and innovation program (Agreement no. 670896). This paper 
was presented, in part, at the IEEE International Conference on Communications, Paris, France, 2017. }
\thanks{Junting Chen is with the School of Science and Engineering and 
the Future Network of Intelligence Institute (FNii) at 
the Chinese University of Hong Kong, Shenzhen, Guangdong 518172, China (e-mail: juntingc@cuhk.edu.cn).}
\thanks{David Gesbert is with Department of Communication Systems, 
EURECOM, 06410 Sophia-Antipolis, France (e-mail: gesbert@eurecom.fr).}
}


\usepackage[acronym]{glossaries}
\newcommand{\newac}{\newacronym}
\newcommand{\ac}{\gls}
\newcommand{\Ac}{\Gls}
\newcommand{\acpl}{\glspl}

\makeglossaries

\newac{speb}{SPEB}{square position error bound}
\newac[plural=EFIMs,firstplural=Fisher information matrices (EFIMs)]{efim}{EFIM}{Fisher information matrix}
\newac{ne}{NE}{Nash equilibrium}
\newac{mse}{MSE}{mean squared error}
\newac{toa}{TOA}{time-of-arrival}
\newac{snr}{SNR}{signal-to-noise ratio}
\newac{lan}{LAN}{local area network}
\newac{psd}{PSD}{positive semidefinite}
\newac{pd}{PD}{positive definite}
\newac{wrt}{w.r.t.}{with respect to}
\newac{lhs}{L.H.S.}{left hand side}
\newac{wp1}{w.p.1}{with probability 1}
\newac{kkt}{KKT}{Karush-Kuhn-Tucker}
\newac{wlog}{w.l.o.g.}{without loss of generality}
\newac{mle}{MLE}{maximum likelihood estimation}
\newac{gps}{GPS}{global positioning system}
\newac{rssi}{RSSI}{received signal strength}
\newac{mimo}{MIMO}{multiple-input multiple-output}
\newac{csi}{CSI}{channel state information}
\newac{fdd}{FDD}{frequency division duplexing}
\newac{ms}{MS}{mobile station}
\newac{bs}{BS}{base station}
\newac{d2d}{D2D}{device-to-device}
\newac{slnr}{SLNR}{signal-to-interference-leakage-and-noise-ratio}
\newac{ula}{ULA}{uniform linear antenna array}
\newac{pas}{PAS}{power angular spectrum}
\newac{mmse}{MMSE}{minimum mean square error}
\newac{zf}{ZF}{zero-forcing}
\newac{rzf}{RZF}{regularized zero-forcing}
\newac{as}{AS}{angular spread}
\newac{aod}{AOD}{angle of departure}
\newac{iid}{i.i.d.}{independent and identically distributed} 
\newac{sinr}{SINR}{signal-to-interference-and-noise ratio}
\newac{tdd}{TDD}{time-division duplex}
\newac{rvq}{RVQ}{random vector quantization}
\newac{rhs}{R.H.S.}{right hand side}
\newac{mrc}{MRC}{maximum ratio combining}
\newac{cdf}{CDF}{cumulative distribution function}
\newac{a.s.}{a.s.}{almost surely}
\newac{los}{LOS}{line-of-sight}
\newac{jsdm}{JSDM}{joint spatial division and multiplexing}
\newac{map}{MAP}{maximum a posteriori}
\newac{klt}{KLT}{Karhunen-Lo\`eve Transform}
\newac{lbe}{LBE}{link bargaining equilibrium}
\newac{se}{SE}{Stackelberg equilibrium}
\newac{uav}{UAV}{unmanned aerial vehicle}
\newac{nlos}{NLOS}{non-line-of-sight}
\newac{pdf}{PDF}{probability density function}
\newac{em}{EM}{expectation-maximization}
\newac{knn}{KNN}{$k$-nearest neighbor}
\newac{svd}{SVD}{singular value decomposition}
\newac{nmf}{NMF}{non-negative matrix factorization}
\newac{umf}{UMF}{unimodality-constrained matrix factorization}
\newac{rmse}{RMSE}{rooted mean squared error}
\newac{olos}{OLOS}{obstructed line-of-sight}
\newac{mmw}{mmW}{millimeter wave}
\newac{ber}{BER}{bit error rate}
\newac{rss}{RSS}{received signal strength}
\newac{lp}{LP}{linear program}
\newac{ufw}{U-FW}{unimodal Frank-Wolfe}
\newac{utf}{UTF}{unimodality-constrained tensor factorization}
\newac{fw}{FW}{Frank-Wolfe}

\newtheorem*{thm1a}{Theorem 1A}

\makeatother

\usepackage{babel}
\providecommand{\conditionname}{Condition}
\providecommand{\definitionname}{Definition}
\providecommand{\lemmaname}{Lemma}
\providecommand{\propositionname}{Proposition}
\providecommand{\theoremname}{Theorem}

\begin{document}
\title{Efficient Local Map Search Algorithms for the Placement of Flying
Relays}

\maketitle
%
%


\newcommand*{\SINGLECOLUMN}{}

\ifdefined\SINGLECOLUMN
	\setkeys{Gin}{width=0.5\columnwidth}
	\newcommand{\figfontsize}{\footnotesize} 
	\newcommand{\labelfontsize}{0.8}
	\newcommand{\ticklabelfontsize}{0.8}
	\newcommand{\numberfontsize}{0.6}
\else
	\setkeys{Gin}{width=1.0\columnwidth}
	\newcommand{\figfontsize}{\normalsize} 
	\newcommand{\labelfontsize}{0.7}
	\newcommand{\ticklabelfontsize}{0.7}
	\newcommand{\numberfontsize}{0.5}
\fi
\begin{abstract}
This paper studies the optimal \ac{uav} placement problem for wireless
networking. The \ac{uav} operates as a flying wireless relay to provide
coverage extension for a \ac{bs} and deliver capacity boost to a
user shadowed by obstacles.
While existing methods rely on statistical models for potential blockage
of a direct propagation link, we propose an approach capable of leveraging
local terrain information to offer performance guarantees. The proposed
method allows to strike the best trade-off between minimizing propagation
distances to ground terminals and discovering good propagation conditions.
The algorithm only requires several propagation parameters, but it
is capable to avoid deep propagation shadowing and is proven to find
the globally optimal UAV position. Only a \emph{local} exploration
over the target area is required, and the maximum length of search
trajectory is linear to the geographical scale. Hence, it lends itself
to online search. Significant throughput gains are found when compared
to other positioning approaches based on statistical propagation models. 
\end{abstract}

\section{Introduction}

\label{sec:intro}

One significant challenge for wireless communication networks is the
rapid increase of demand for high data rate and low latency wireless
service. As a promising solution to future communication networks,
substantial attention has been brought on the exploitation of \acpl{uav}
as flying relays to connect \acpl{bs} with the users in communication
outage \cite{ZenZhaLim:J16,MozSaaBenDeb:J16c,XiaXiaXia:M16,VanChiPol:J16,MerGuv:C15,HoSujJohDe:C13,sharma2016uav}. 

In UAV relaying over a dense urban environment, a fundamental challenge
is the shadowing at the user side, where the degree of obstruction
depends on the geographical environment. For example, the link between
a UAV and a user may be in deep shadow when the UAV is located at
the east side of a building, whereas, the propagation condition may
be significantly improved when the UAV moves to the north side. It
is difficult to know such fine-grained propagation condition prior
to physically flying a UAV to a target position for assessment. Existing
techniques that model the possible obstruction include specifying
a larger path loss exponent, adding additional power loss, and constructing
a random variable that describes the statistics of the shadowing.
However, these models still over-simplify the actual propagation,
because they implicitly assume that the degree of obstruction is \emph{homogeneous}
everywhere, meaning that the path loss is statistically the same everywhere
given the same propagation distance (and elevation angle). Consequently,
in a BS-UAV-user relay network, these models would predict the best
UAV relay on the BS-user axis. In reality, however, a UAV may find
significantly better propagation condition \emph{off} the BS-user
axis. We will demonstrate that substantial performance gain can be
achieved when a more realistic fine-grained propagation model is exploited.
The main goal of the paper is to develop an efficient and blockage-adaptive
search strategy to explore fine-grained propagation condition for
the optimal UAV position.

\subsection{Related Works}

UAV optimization strategies are mainly developed based on specific
air-to-ground path loss models.  In \cite{OuyZhuLinLiu:J14,ZenZha:J17,JinZhaCha:C12,ChoJunSun:C14,ChaChaGagGag:C14,JiaSwi:J12,AzaRosChePol:J18,ur2018uav,lyu2017placement},
the path loss was modeled as a deterministic function of the UAV-to-user
distance, irrelevant to specific UAV positions. Thanks to the simplicity
of these distance-based models, \cite{OuyZhuLinLiu:J14,ZenZha:J17,JinZhaCha:C12,ChoJunSun:C14,ChaChaGagGag:C14}
developed solutions to UAV navigation problems, \cite{JiaSwi:J12}
studied optimization strategies for \ac{mimo} communications with
UAVs with multiple antennas, and \cite{AzaRosChePol:J18} optimized
the UAV position for cooperative communications. The models used in
\cite{OuyZhuLinLiu:J14,ZenZha:J17,JinZhaCha:C12,ChoJunSun:C14,ChaChaGagGag:C14,JiaSwi:J12,AzaRosChePol:J18,ur2018uav,lyu2017placement}
imply that the path loss is the same under the same distance, but
more detailed research in \cite{AlhKanJam:C14,MozSaaBenDeb:C15} suggests
that air-to-ground propagation should also depend on the elevation
angle of the UAV-user link. 

To capture the dependency of obstruction on elevation angles, in \cite{AlhKanJam:C14,MozSaaBenDeb:C15,HouSitLar:J14,MozSaaBenDeb:C16,MozSaaBenDeb:J16b},
the path loss was modeled as a random variable of the propagation
distance and the elevation angle. Specifically, for the UAV position
and coverage optimization problems studied in \cite{MozSaaBenDeb:C15,HouSitLar:J14,MozSaaBenDeb:C16},
the path loss was modeled as the average of the path loss under the
LOS case and that under the NLOS case, where the larger the elevation
angle, the higher the LOS probability. The work \cite{MozSaaBenDeb:J16b}
considered the shadowing statistics as a function of elevation angle.
Yet, these models implicitly assume that the degree of obstruction
is homogeneous for the same distance and elevation angle. In practice,
however, the degree of obstruction may vary from one location to another. 

\subsection{Challenges and Our Contributions}

Although it is theoretically possible to associate each UAV and user
position pair with a channel quality to capture the variation of the
terrain, it is almost infeasible to implement a search strategy for
the optimal UAV position. On one hand, it is time and energy prohibitive
to maneuver a UAV for channel quality assessment at every possible
position. On the other hand, even a 3D city map or a radio map can
be constructed, it is computationally expensive to exhaustively search
of the optimal UAV position. Moreover, the city map or radio map may
not be updated or reliable in practice. 

Therefore, it is essential to develop an air-to-ground path loss model
to capture the variation of the terrain, and at the same time, consists
of some nice structure that can be exploited to efficiently search
for the optimal UAV position. Specifically, two main issues are to
be addressed:
\begin{itemize}
\item \textbf{How to model the air-to-ground path loss for the communication
between ground users and low-altitude UAVs?} We need a model that
not only captures the various degrees of obstruction due to the complex
terrain, but also facilitates a low complexity search strategy for
optimal UAV positioning. 
\item \textbf{How to plan an efficient search trajectory for the optimal
UAV position?} We desire to build a search path to find the globally
optimal UAV position, while the maximum search length is only linear
to the radius of the search area.
\end{itemize}

To answer these questions, we develop a \emph{nested segmented} air-to-ground
propagation model, which, for each user position, partitions the UAV
search area into several segments, and associates each segment with
a path loss model, such as LOS model, obstructed LOS model and \ac{nlos}
model. This feature matches with the observation obtained from the
experimental study in \cite{FenGeeTamNix:C06}. In addition, we impose
that as the UAV moves away from the user, it can only enter propagation
segments with a higher degree of obstruction, as shown in Fig. \ref{fig:geometry-property}
(b). Such a requirement is consistent with many existing air-to-ground
models from both academia and industry \cite{AlhKanJam:C14,TR36777},
where the lower the elevation angle, the higher the probability for
which the UAV-user link is obstructed. Using such a model, we develop
a \emph{shaded-contour-exploration }algorithm to search for the \emph{globally}
optimal UAV position, and the algorithm only requires several propagation
parameters but not the entire radio map. We prove that global optimum
can be attained under a linear search trajectory. While this idea
was partially exploited in our preliminary work \cite{CheGes:C17},
the work \cite{CheGes:C17} focused only on the simplest case of two
propagation segments and the optimality proof was not presented. 

To summarize, the key contributions are made as follows:
\begin{itemize}
\item We propose a nested segmented propagation model to capture the fine-grained
degree of obstruction in air-to-ground propagation. 
\item We develop a shaded-contour-exploration strategy to find the optimal
UAV position for a single user relay system, with proven global optimality
and linear search complexity. Substantial performance gain for the
case of clustered multiple users is also numerically demonstrated. 
\item We perform numerical experiments to evaluate the performance of the
UAV relay system and compare with existing approaches from existing
models. Substantial throughput gain is found from a simulated Manhattan-like
urban environment.
\end{itemize}

The rest of the paper is organized as follows. Section \ref{sec:system-model-dual}
establishes the nested segmented air-to-ground propagation model.
Section \ref{sec:framework-property} develops the search algorithm,
and Section \ref{sec:global-convergence} establishes theoretical
results on the global optimality. Numerical results are demonstrated
in Section \ref{sec:numerical}, and conclusions are given in Section
\ref{sec:conclusion}.

\section{System Model}

\label{sec:system-model-dual}

Consider a cellular network in an urban environment, where the BS
is placed on rooftop or on a tower that is higher than all the buildings.
Due to the possibly dense distribution of buildings and vegetation,
it is likely that signals transmitted from the BS are significantly
obstructed from users on street levels. Consider to deploy a UAV as
a flying relay that connects a user with the BS. Assume that the UAV
moves at a \emph{fixed} height $H_{\text{d}}>H_{\text{b}}$, where
$H_{\text{b}}$ is the height of the BS. Denote the horizontal positions
of the UAV, BS, and user, respectively, as $\mathbf{x},\mathbf{x}_{\text{b}},\mathbf{x}_{\text{u}}\in\mathbb{R}^{2}$,
and hence, $(\mathbf{x},H_{\text{d}}),(\mathbf{x}_{\text{b}},H_{\text{b}}),(\mathbf{x}_{\text{u}},0)\in\mathbb{R}^{3}$
are, respectively, the positions of the UAV, BS, and user in 3D. 

Note that the signal from the UAV can still be obstructed from the
user due to local obstacles surrounding the user. On the other hand,
the UAV relay cannot move too close to the user because it needs to
balance the relay link with the BS. To address this dilemma, the goal
of this paper is to optimize the horizontal position $\mathbf{x}\in\mathbb{R}^{2}$
for the UAV. 

\subsection{Channel Model for the UAV-BS Link}

Define the \emph{channel} as the deterministic power gain averaged
over small scale fading. The BS-UAV channel is modeled as 
\begin{equation}
g_{\text{b}}(\mathbf{x})=\beta_{0}d_{\text{b}}(\mathbf{x})^{-\alpha_{0}}\label{eq:UAV-BS-channel-model}
\end{equation}
where $d_{\text{b}}(\mathbf{x})=\sqrt{\|\mathbf{x}-\mathbf{x}_{\text{b}}\|^{2}+(H_{\text{d}}-H_{\text{b}})^{2}}$
is the distance from the UAV at $(\mathbf{x},H_{\text{d}})$ to the
BS at $(\mathbf{x}_{\text{b}},H_{\text{b}})$, and the constants $\alpha_{0}>1$
and $\beta_{0}>0$ are the classical path loss exponent and offset
parameters. Such a model corresponds to the common scenario where
the altitudes $H_{\text{b}}$ of the BS and $H_{\text{d}}$ of the
UAV are large enough such that there is \emph{always} a LOS condition
for the BS-UAV link. We thus focus on modeling the UAV-user link in
Section \ref{subsec:nested-segmented-model}.

\subsection{Channel Model for the UAV-user Link: A Nested Segmented Model}

\label{subsec:nested-segmented-model}

\begin{figure}
\begin{centering}
\subfigure[Horizontal view]{\includegraphics[width=0.85\columnwidth]{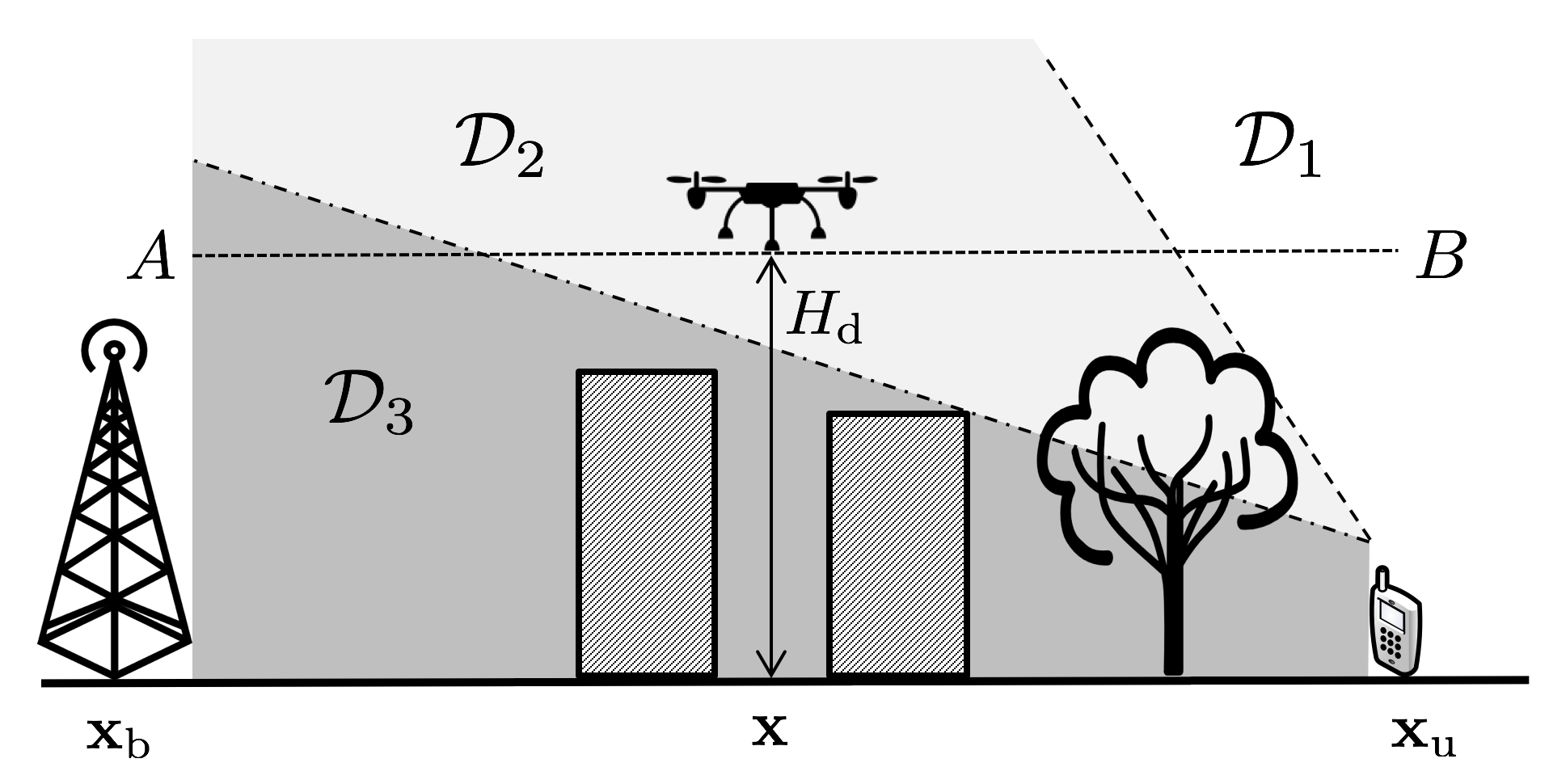}}
\par\end{centering}
\begin{centering}
\subfigure[Nested segmented model]{\includegraphics[width=0.7\columnwidth]{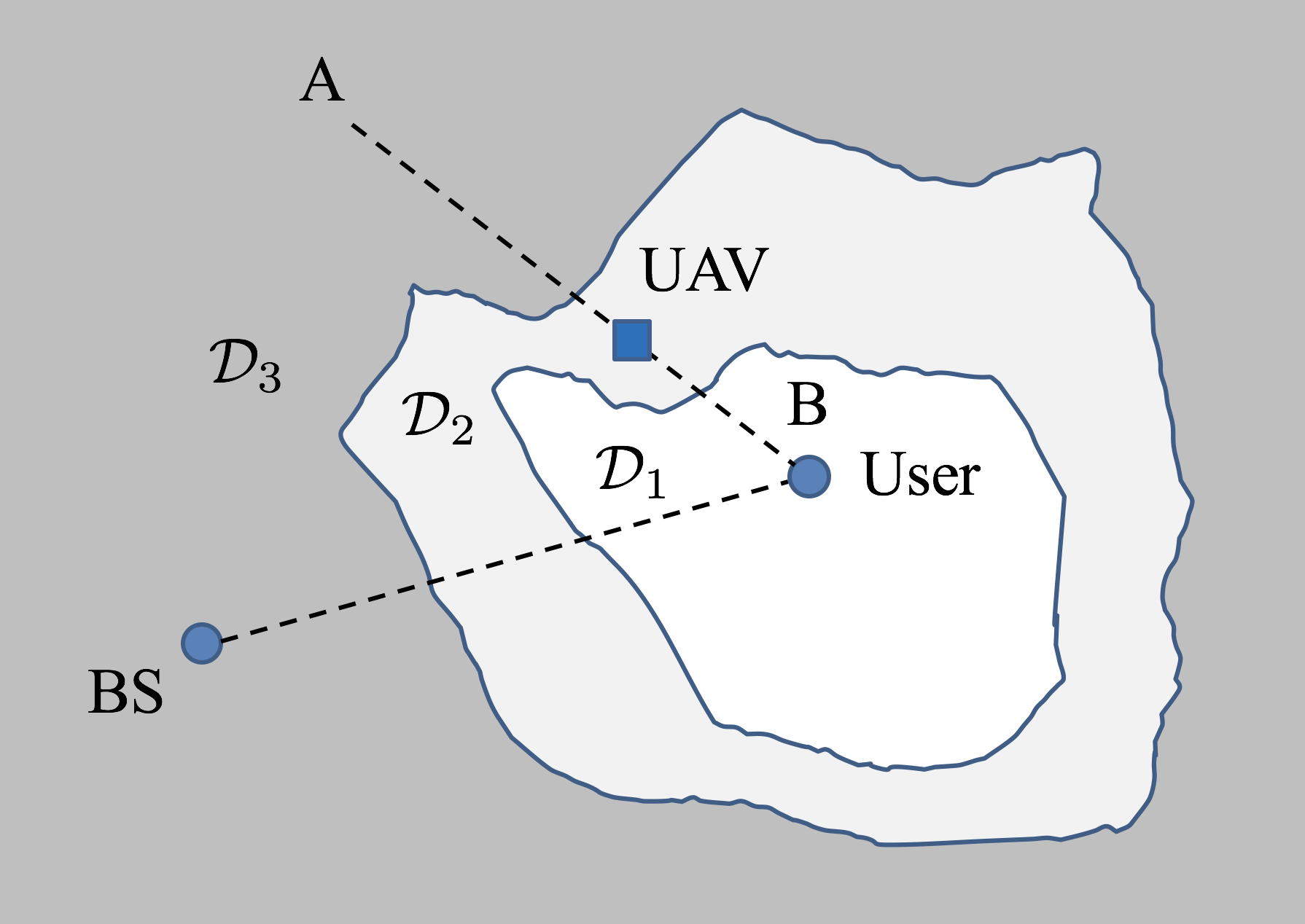}} 
\par\end{centering}
\caption{\label{fig:geometry-property} A geometric interpretation of the nested
segmented model from (a) horizontal view and (b) top view.}
\end{figure}
In the classical large-scale fading channel model, the channel gain
is modeled as $G_{\text{dB}}=b-a\log_{10}d+\xi$, where $\xi$ is
a random variable to capture the shadowing effect. Inspired by the
geometry in the ray-tracing propagation modeling as illustrated in
Fig. \ref{fig:geometry-property}, we propose to split $G_{\text{dB}}$
into several components each associated with a set of parameters $a,b$
and a random component $\xi_{k}$ for a specific degree of link obstruction. 

Specifically, let $\mathbb{D}\subseteq\mathbb{R}^{2}$ be the domain
of all possible UAV positions $\mathbf{x}$ at constant altitude $H_{\text{d}}$.
Consider a partition of $\mathbb{D}$ into $K$ disjoint segments
$\mathbb{D}=\mathcal{D}_{1}(\tilde{\mathbf{x}}_{\text{u}})\cup\mathcal{D}_{2}(\tilde{\mathbf{x}}_{\text{u}})\cup\dots\cup\mathcal{D}_{K}(\tilde{\mathbf{x}}_{\text{u}})$,
where $\mathcal{D}_{k}\cap\mathcal{D}_{j}=\varnothing$, for $k\neq j$,
and $\mathcal{D}_{k}(\tilde{\mathbf{x}}_{\text{u}})$ denotes the
region of UAV locations for which the UAV maintains a degree-$k$
of LOS obstruction from the user. The proposed \emph{segmented} propagation
model for $G_{\text{dB}}$ is specified as: 
\begin{equation}
G_{\text{dB}}(\mathbf{x})=\sum_{k=1}^{K}\big(b_{k}-a_{k}\log_{10}d_{\text{u}}(\mathbf{x})+\xi_{k}\big)\mathbb{I}\{\mathbf{x}\in\mathcal{D}_{k}\}\label{eq:uav-user-channle-model-with-noise}
\end{equation}
where $d_{\text{u}}(\mathbf{x})=\sqrt{\|\mathbf{x}-\mathbf{x}_{\text{u}}\|^{2}+(H_{\text{d}}-H_{\text{u}})^{2}}$
is the distance from the UAV located at $(\mathbf{x},H_{\text{d}})$
to the user at $\tilde{\mathbf{x}}_{\text{u}}=(\mathbf{x}_{\text{u}},H_{\text{u}})$,
$a_{k}$ and $b_{k}$ are some parameters, and $\mathbb{I}\{A\}$
is an indicator function taking value $1$ if condition $A$ is satisfied,
and $0$ otherwise. The random variable $\xi_{k}$ captures the residual
shadowing effect. The segment parameters $\{\alpha_{k},\beta_{k},\mathcal{D}_{k}\}$
are assumed to satisfy the following conditions:
\begin{enumerate}
\item The propagation segment $\mathcal{D}_{k}$ exhibits a higher degree
of LOS obstruction than $\mathcal{D}_{k-1}$, \emph{i.e.}, for $k=2,3,\dots,K$
and any UAV position $\mathbf{x}$, 
\[
b_{k}-a_{k}\log_{10}d_{\text{u}}(\mathbf{x})<b_{k-1}-a_{k-1}\log_{10}d_{\text{u}}(\mathbf{x})
\]
or, in the linear scale representation,
\begin{equation}
\beta_{k}d_{\text{u}}(\mathbf{x})^{-\alpha_{k}}<\beta_{k-1}d_{\text{u}}(\mathbf{x})^{-\alpha_{k-1}}\label{eq:nested-segmented-model-order-condition}
\end{equation}
where $b_{k}=10\log_{10}\beta_{k}$ and $a_{k}=10\alpha_{k}$. 
\item The propagation segments $\mathcal{D}_{k}$ are \emph{nested} along
any directions from the user, \emph{i.e.}, for any $\mathbf{x}\in\mathcal{D}_{k}$
and $0\leq\rho\leq1$ 
\begin{equation}
\mathbf{x}_{\text{u}}+\rho(\mathbf{x}-\mathbf{x}_{\text{u}})\in\mathcal{D}_{j},\qquad\text{for some }j\leq k.\label{eq:nested-segmented-model-nested-condition}
\end{equation}
In other words, when the UAV moves towards the user, the UAV-user
channel tends to become \emph{always} less obstructed \cite{BaiHea:J15}
as there are fewer obstacles in between the UAV and the user, as illustrated
in Fig. \ref{fig:geometry-property}. 
\end{enumerate}

It is intuitive that if the parameters $\{\alpha_{k},\beta_{k},\mathcal{D}_{k}\}$
are obtained correctly, the variance of the residual term $\xi_{k}$
can be substantially reduced compared to that of the random component
$\xi$ in the classical channel model. For example, the variance of
the shadowing in the LOS case is believed to be much smaller than
that in the combined LOS and NLOS case. 

With such an insight, we focus on optimizing the UAV network based
on the average channel gain $\bar{G}_{\text{dB}}(\mathbf{x})\triangleq\mathbb{E}\left\{ G_{\text{dB}}(\mathbf{x})\big|\{\mathcal{D}_{k}\}\right\} $
given the propagation segments, where the expectation is taken over
$\xi_{k}$. Therefore, assuming $\xi_{k}$ with zero mean, the deterministic
channel in linear scale $10\log_{10}g_{\text{u}}(\mathbf{x})\triangleq\bar{G}_{\text{dB}}(\mathbf{x})$
can be written as: 
\begin{equation}
g_{\text{u}}(\mathbf{x})=\sum_{k=1}^{K}\beta_{k}d_{\text{u}}(\mathbf{x})^{-\alpha_{k}}\mathbb{I}\{\mathbf{x}\in\mathcal{D}_{k}\}.\label{eq:uav-user-channel-model}
\end{equation}

The problem of learning $\alpha_{k}$, $\beta_{k}$, and the statistics
of $\xi_{k}$ from measurement data has been partially attempted in
\cite{CheYanGes:C16,CheEsrGesMit:C17}, and it is not the focus of
this paper. Here, we assume that $\alpha_{k}$, $\beta_{k}$, and
the statistics of $\xi_{k}$ are perfectly known, but we still need
to (partially) determine $\mathcal{D}_{k}$, \emph{i.e.}, the boundaries
shown in Fig. \ref{fig:geometry-property} (b), to help search for
the optimal UAV position. Note that learning the entire $\mathcal{D}_{k}$
is much more difficult and time consuming (or even prohibitive) than
learning $\alpha_{k}$ and $\beta_{k}$. The goal of this paper is
to optimize the UAV position by only \emph{partially} exploring $\mathcal{D}_{k}$.

\subsection{Problem Formulation and Application Examples}

\label{subsec:problem-formulation}

Consider an objective cost function $f(g_{\text{u}},g_{\text{b}})$
of the UAV-user channel gain $g_{\text{u}}$ and the BS-UAV channel
gain $g_{\text{b}}$. Assume that $f(x,y)$ is a continuous and decreasing
function in $x$ and $y$, respectively. A generic UAV positioning
problem can be formulated as follows\footnote{More generally, one can design the cost function $f$ in terms of
the UAV position $\mathbf{x}$ to capture other system characteristics,
such as antenna pattern and effect of small-scale fading.}
\[
\mathscr{P}:\qquad\underset{\mathbf{x}\in\mathbb{R}^{2}}{\text{minimize}}\quad f(g_{\text{u}}(\mathbf{x}),g_{\text{b}}(\mathbf{x})).
\]

The problem formulation $\mathscr{P}$ can capture many applications
for a variety of relay transmission strategies. Three examples are
illustrated as follows, where we choose some simple formulations for
easy demonstration of the proposed algorithm.

Consider that the transmission from the BS to the UAV is modeled as
$y_{\text{r}}=\sqrt{P_{\text{b}}g_{\text{b}}}a_{\text{b}}s+n_{\text{r}}$,
and that from the UAV to the user is modeled as $y_{\text{u}}=\sqrt{P_{\text{u}}g_{\text{u}}}a_{\text{u}}s_{\text{r}}+n_{\text{u}}$,
where $n_{\text{r}},n_{\text{u}}\sim\mathcal{CN}(0,1)$ are the receive
noise at the UAV relay and the user, respectively, $s,s_{\text{r}}\sim\mathcal{N}(0,1)$
are the transmit signals from the BS and the UAV relay, respectively,
and $P_{\text{u}}$ and $P_{\text{b}}$ are transmission powers at
the UAV and the BS, respectively. The variables $a_{\text{b}}$ and
$a_{\text{u}}$ model the small scale fading on the BS-UAV link and
the UAV-user link, respectively. For Rayleigh fading channels, $|a_{\text{b}}|^{2}$
and $|a_{\text{u}}|^{2}$ are assumed to follow exponential distribution
with parameter (normalized as) $\lambda=1$. 

\subsubsection{Amplify-and-Forward}

\label{subsec:example-AF}

\begin{figure*}
\begin{centering}
\subfigure[Urban topology]{
\psfragscanon
\psfrag{BS}[][][0.5]{BS}
\psfrag{UE 1}[Bl][Bl][0.5]{User}
\psfrag{Building height [m]}[][][0.5]{Building height [m]}
\psfrag{X-axis}[][][0.5]{Longitude [m]}
\psfrag{Y-axis}[][][0.5]{Latitude [m]}
\psfrag{0}[][][0.4]{0}
\psfrag{15}[][][0.4]{15}
\psfrag{25}[][][0.4]{25}
\psfrag{35}[][][0.4]{35}
\psfrag{45}[][][0.4]{45}
\psfrag{200}[][][0.4]{200}
\psfrag{400}[][][0.4]{400}
\psfrag{600}[][][0.4]{600}
\psfrag{800}[][][0.4]{800}
\psfrag{1000}[][][0.4]{1000}
\psfrag{9}[][][0.4]{9}
\psfrag{18}[][][0.4]{18}
\psfrag{27}[][][0.4]{27}
\psfrag{36}[][][0.4]{36}
\psfrag{45}[][][0.4]{45}\includegraphics[width=0.66\columnwidth]{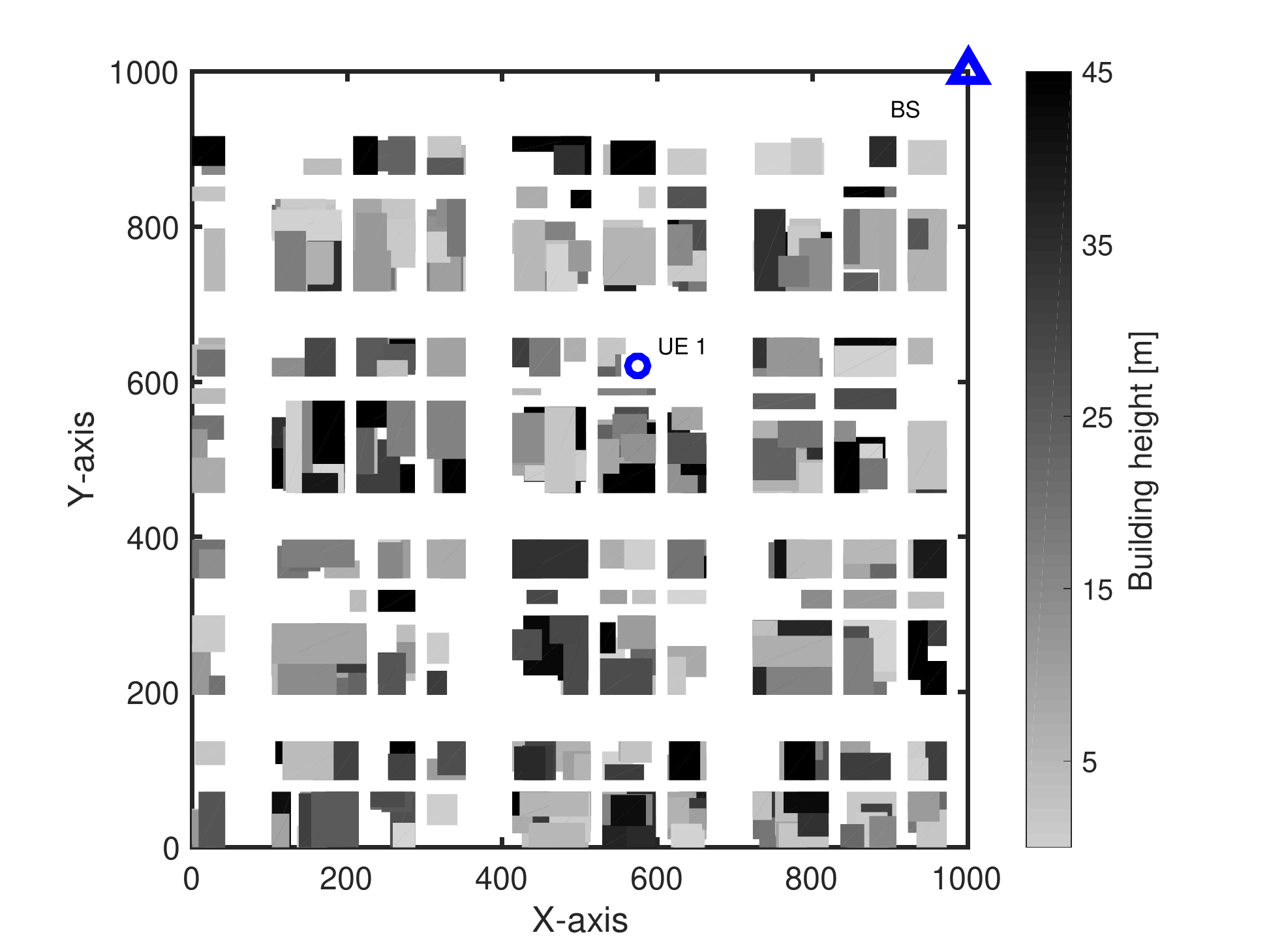}}\subfigure[Power map]{
\psfragscanon
\psfrag{Received power (dBm)}[][][0.5]{Received power (dBm)}
\psfrag{X-axis (m)}[][][0.6]{X-axis (m)}
\psfrag{Y-axis (m)}[][][0.6]{Y-axis (m)}
\psfrag{0}[][][0.5]{0}
\psfrag{200}[][][0.5]{200}
\psfrag{400}[][][0.5]{400}
\psfrag{600}[][][0.5]{600}
\psfrag{800}[][][0.5]{800}
\psfrag{1000}[][][0.5]{1000}\includegraphics[width=0.66\columnwidth]{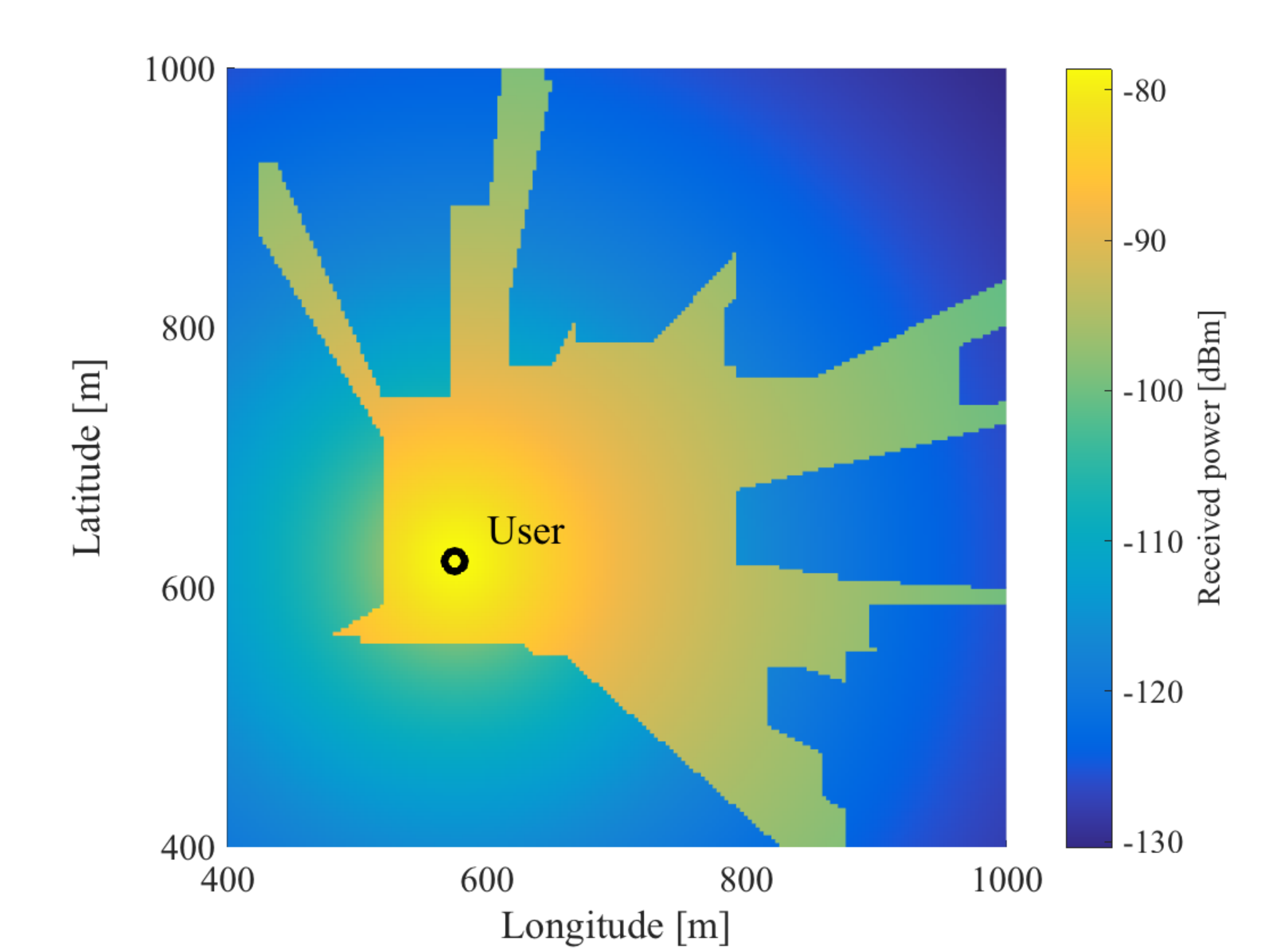}}\subfigure[Capacity map]{
\psfragscanon
\psfrag{Received power (dBm)}[][][0.5]{Received power (dBm)}
\psfrag{X-axis (m)}[][][0.6]{X-axis (m)}
\psfrag{Y-axis (m)}[][][0.6]{Y-axis (m)}
\psfrag{0}[][][0.5]{0}
\psfrag{200}[][][0.5]{200}
\psfrag{400}[][][0.5]{400}
\psfrag{600}[][][0.5]{600}
\psfrag{800}[][][0.5]{800}
\psfrag{1000}[][][0.5]{1000}\includegraphics[width=0.66\columnwidth]{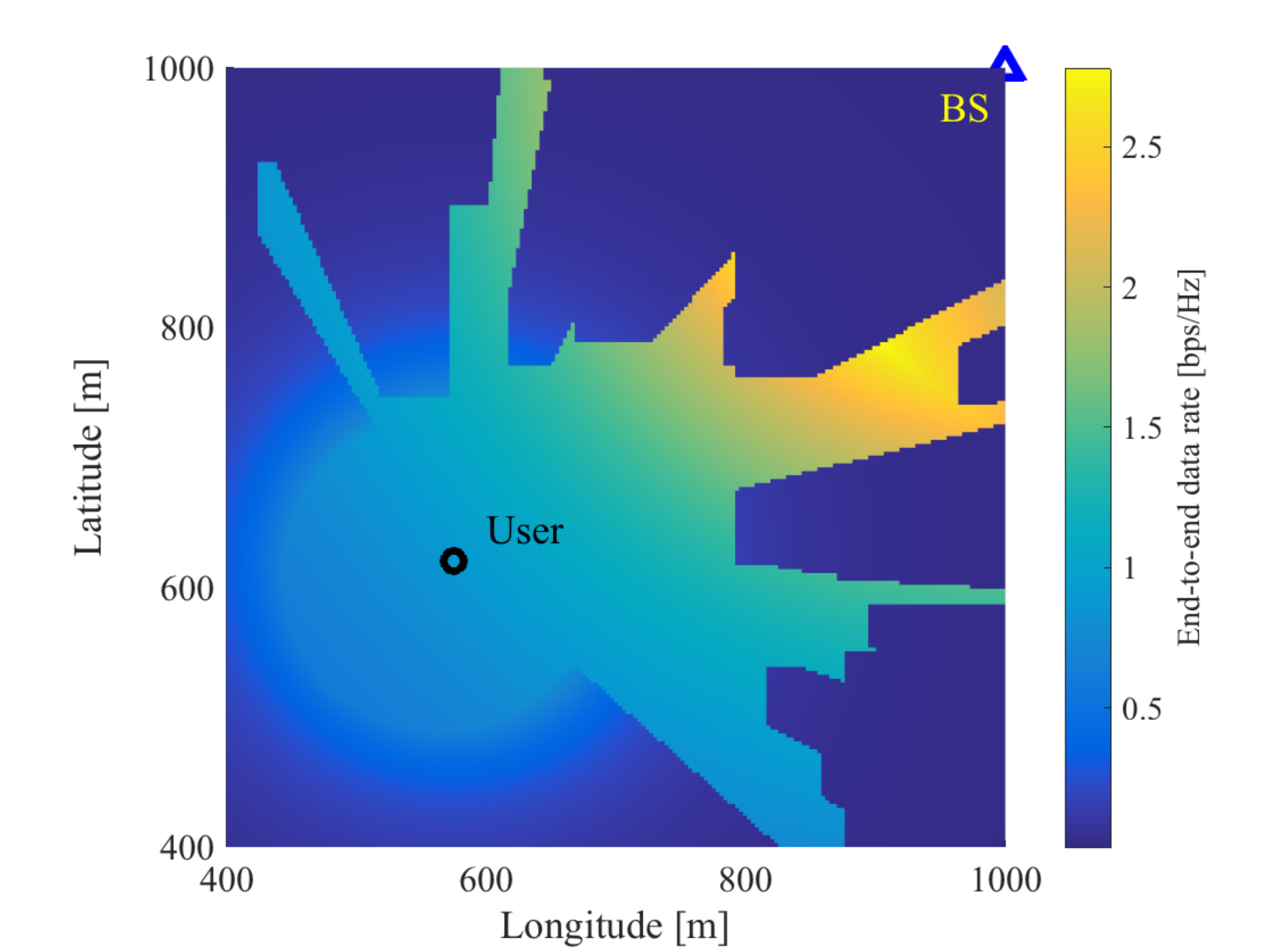}}
\par\end{centering}
\caption{\label{fig:urban-city-map} (a) A map of a dense urban area, where
the rectangles denote the building with colors representing their
heights. (b) The simulated received power map corresponding to every
UAV position. (c) The simulated end-to-end capacity map.}
\end{figure*}
In amplify-and-forward relaying, the UAV relays the information by
transmitting $s_{\text{r}}=y_{\text{r}}/\sqrt{P_{\text{b}}g_{\text{b}}|a_{\text{r}}|^{2}+1}$,
where the scaling factor $\sqrt{P_{\text{b}}g_{\text{b}}|a_{\text{r}}|^{2}+1}$
is to normalize the transmission power at the UAV to be $P_{\text{u}}$.
It was shown in \cite{NabBolKne:J04,LanTseWor:J04} that the capacity
of the relay channel is given by $C_{\text{AF}}=\frac{1}{2}\log_{2}\big(1+q(P_{\text{b}}g_{\text{b}}|a_{\text{r}}|^{2},P_{\text{u}}g_{\text{u}}|a_{\text{u}}|^{2})\big)$,
where $q(x,y)\triangleq xy/(x+y+1)$ and the parameter $\frac{1}{2}$
is to capture the fact that the information requires two time slots
to reach the user. The outage probability with respect to a target
data rate $R$ was shown to be \cite[Lemma 1]{LanTseWor:J04} $\mathbb{P}\{C_{\text{AF}}<R\}\approx\Big(\frac{1}{P_{\text{b}}g_{\text{b}}}+\frac{1}{P_{\text{u}}g_{\text{u}}}\Big)(2^{2R}-1)^{2}$
under high \ac{snr}, \emph{i.e.}, $P_{\text{b}}g_{\text{b}},P_{\text{u}}g_{\text{u}}\gg1$.\footnote{The original problem in \cite{LanTseWor:J04} considered a diversity
scheme that combines the signal from the relay and the signal from
the BS. Such a strategy also leads to problem $\mathscr{P}1$ under
high \ac{snr}. } 

Therefore, to minimize the outage probability of the relay channel,
the desired UAV position can be determined as the solution to $\mathscr{P}$
with a cost function given as: 
\begin{align}
 & f(g_{\text{u}}(\mathbf{x}),g_{\text{b}}(\mathbf{x})):=\frac{1}{P_{\text{u}}g_{\text{u}}(\mathbf{x})}+\frac{1}{P_{\text{b}}g_{\text{b}}(\mathbf{x})}.\label{eq:outage-minimization}
\end{align}

\subsubsection{Decode-and-Forward}

\label{subsec:example-DF}

In decode-and-forward relaying, the UAV fully decodes the message
$\hat{s}$ from the receive signal $y_{\text{r}}$, and transmits
$s_{\text{r}}=\hat{s}$ to the user. The maximum capacity of such
a decode-and-forward relay system can be shown to be $C_{\text{DF}}=\frac{1}{2}\min\{\log_{2}(1+P_{\text{b}}g_{b}|a_{b}|^{2}),\log_{2}(1+P_{\text{u}}g_{\text{u}}|a_{\text{u}}|^{2})\}$
\cite{LanTseWor:J04,WanCanGiaLan:J07}. Using Jensen's inequality
$\mathbb{E}\{C(x)\}\leq C(\mathbb{E}\{x\})$ on a concave function
$C(x)$, an upper bound of the ergodic capacity $\mathbb{E}\{C_{\text{DF}}\}$
is given by $\frac{1}{2}\min\{\log_{2}(1+P_{\text{b}}g_{b}),\log_{2}(1+P_{\text{u}}g_{\text{u}})\}$.
The desired UAV position can be determined by maximizing such a capacity
bound. Equivalently, the problem $\mathscr{P}$ can be specified by
choosing the following cost function: 
\begin{align}
f(g_{\text{u}}(\mathbf{x}),g_{\text{b}}(\mathbf{x})):=\max\Big\{ & -\log_{2}\big(1+P_{\text{b}}g_{b}(\mathbf{x})\big),\label{eq:drone-rate-maximization}\\
 & -\log_{2}\big(1+P_{\text{u}}g_{\text{u}}(\mathbf{x})\big)\Big\}.\nonumber 
\end{align}

A numerical example is given in Fig. \ref{fig:urban-city-map}, where
Fig. \ref{fig:urban-city-map} (b) simulates the received power of
the UAV-user signal \ac{wrt} every UAV position under a segmented
propagation model with $K=2$ segments. Fig. \ref{fig:urban-city-map}
(c) shows the corresponding relay channel capacity from the BS to
the user via the UAV. It is not trivial to find the optimal UAV relay
position due to the irregular propagation pattern. 

\subsubsection{Multiuser Clustered around a Hotspot}

\label{subsec:multiuser}

Suppose that there are $N_{\text{u}}$ users clustered around a hotspot
centered at $\mathbf{x_{\text{c}}}$ with radius $r_{\text{u}}$.
Let $f^{(i)}(\mathbf{x})$ be the cost function taking the form in
(\ref{eq:drone-rate-maximization}) for the $i$th user located at
position $\mathbf{x}_{\text{u}}^{(i)}$. Specifically, the UAV-user
gain $g_{\text{u}}(\mathbf{x};\mathbf{x}_{\text{u}}^{(i)})$ in (\ref{eq:uav-user-channel-model})
is computed based on the user position $\mathbf{x}_{i}$. Consider
to maximize the sum rate $-\frac{1}{N_{\text{u}}}\sum_{i=1}^{N_{\text{u}}}f^{(i)}(\mathbf{x})$.
One may consider to simplify and approximate the cost function $\bar{f}(\mathbf{x})\triangleq\frac{1}{N_{\text{u}}}\sum_{i=1}^{N_{\text{u}}}f^{(i)}(\mathbf{x})$
by constructing a virtual user indexed as $i=N_{\text{u}}+1$. The
virtual user is virtually placed at the hotspot center $\mathbf{x}_{\text{c}}$
(as a similar topological model discussed in \cite{ur2018uav}), and
the corresponding channel gain $g_{\text{u}}(\mathbf{x};\mathbf{x}_{\text{c}})$
for is modeled using the segmented log-distance model (\ref{eq:uav-user-channel-model}),
except that the propagation segment depends on the majority vote from
the $N_{\text{u}}$ actual users. Specifically, the UAV position \textbf{$\mathbf{x}$}
belongs to the $k$th propagation segment $\tilde{\mathcal{D}}_{k}(\mathbf{x}_{\text{c}})$
(for the virtual user), if the majority UAV-user links $(\mathbf{x},\mathbf{x}_{\text{u}}^{(i)})$
belong to the $k$the segment. As a result, it is clear that the cost
function $f(\mathbf{x})\triangleq f^{(N_{\text{u}}+1)}(\mathbf{x})$
for the virtual user is a good approximation of the average cost $\frac{1}{N_{\text{u}}}\sum_{i=1}^{N_{\text{u}}}f^{(i)}(\mathbf{x})$
as long as the cluster radius $r_{\text{u}}$ is small. 

While such an approximation is only suboptimal, in Section \ref{subsec:clustered-user-case},
we numerically demonstrate that our proposed strategy that solves
$\mathscr{P}$ still provide reasonably good performance up to moderate
cluster radius $r_{\text{u}}$ as compared to stochastic optimization
using simplified models. 

\section{Algorithm Designs}

\label{sec:framework-property}

In this section, we first derive some useful insights on the optimal
UAV positions, and then develop a polar representation tool. Based
on that, we develop the search algorithm for the optimal UAV position. 

\subsection{Properties of the Optimal UAV Position}
\begin{prop}
\label{prop:Solution-Property} The optimal solution $\mathbf{x}^{*}$
to $\mathscr{P}$ is either on the BS-user axis, or on the boundary
between two propagation segments.
\end{prop}
\begin{proof}
Suppose there is a solution $\mathbf{x}$ which is strictly inside
a propagation segment $\mathcal{D}_{k}$ and is off the BS-user axis.
Then there exists a direction $\bm{\delta}$, such that for a sufficiently
small $\epsilon>0$, the new UAV position $\mathbf{x}+\epsilon\bm{\delta}$
decreases the distances to both the user and the BS, while at the
same time, $\mathbf{x}+\epsilon\bm{\delta}\in\mathcal{D}_{k}$. From
the nested segmented model (\ref{eq:uav-user-channel-model}) and
also the UAV-BS channel model (\ref{eq:UAV-BS-channel-model}), smaller
distances $d_{\text{u}}$ and $d_{\text{b}}$ imply larger channel
gains $g_{\text{u}}$ and $g_{\text{b}}$. Due to the monotonicity
property of the cost function $f(g_{\text{u}},g_{\text{b}})$, larger
channel gains yield a smaller cost value, which implies that $\mathbf{x}$
is not the optimal solution. By contradiction, the proposition is
therefore confirmed. 
\end{proof}

However, it is still costly to search along the segment boundaries
because the boundaries may have complex shapes, resulting in an unacceptably
long trajectory (which could be super-linear to the radius of the
search area) for the UAV as demonstrated in Fig.~\ref{fig:urban-city-map}. 

\subsection{Polar Representation}

\label{subsec:polar-representation}

\begin{figure}
\begin{centering}
\includegraphics[width=0.7\columnwidth]{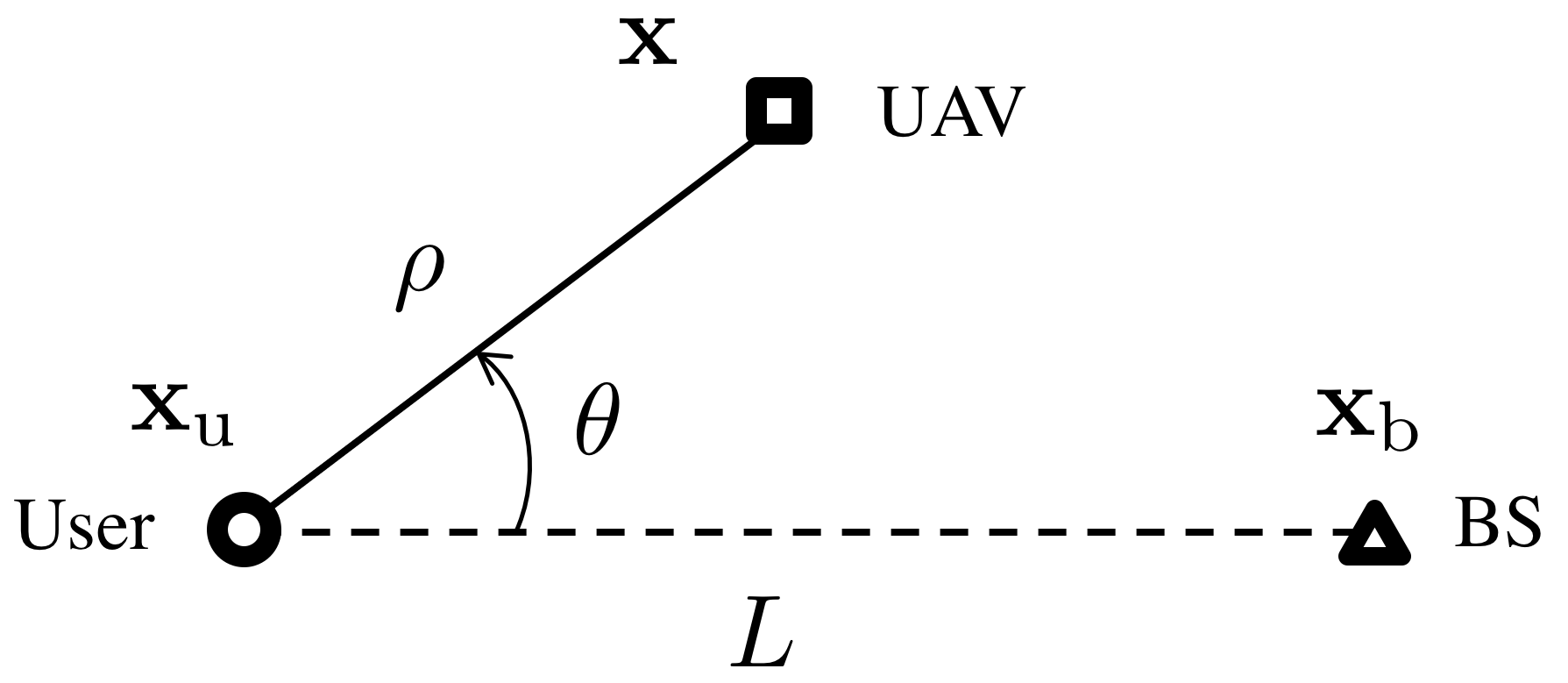}
\par\end{centering}
\caption{\label{fig:polar-coordinate} Illustration of a polar representation
of the UAV position $\mathbf{x}$.}
\end{figure}
To develop a more efficient search strategy, we transform the problem
into the one expressed over a polar coordinate system.

Let $\rho=\|\mathbf{x}-\mathbf{x}_{\text{u}}\|$ be the ground projected
distance from the user at $(\mathbf{x}_{\text{u}},H_{\text{u}})$
to the UAV at $(\mathbf{x},H_{\text{d}})$. Let $\theta\in(-\pi,\pi)$
be the \emph{deviation angle} from the user-to-BS direction to the
user-to-UAV direction as illustrated in Fig.~\ref{fig:polar-coordinate}.
Denote $\mathbf{u}=(u_{1},u_{2})\triangleq\frac{\mathbf{x}_{\text{b}}-\mathbf{x}_{\text{u}}}{\|\mathbf{x}_{\text{b}}-\mathbf{x}_{\text{u}}\|}$
as the normalized user-to-BS direction. The UAV position $\mathbf{x}\in\mathbb{R}^{2}$
can be equivalently expressed by $(\rho,\theta)$ as 
\begin{equation}
\mathbf{x}(\rho,\theta)=\mathbf{x}_{\text{u}}+\rho\mathbf{M}(\theta)\mathbf{u}\label{eq:polar-representation-xy-plane}
\end{equation}
where 
\begin{equation}
\mathbf{M}(\theta)=\left[\begin{array}{ll}
\cos\theta & -\sin\theta\\
\sin\theta & \cos\theta
\end{array}\right]\label{eq:rotation-matrix}
\end{equation}
is a rotation matrix, and 
\begin{equation}
\theta=\text{sign}(z_{2}u_{1}-z_{1}u_{2})\cdot\text{arccos}\left(\mathbf{z}^{\text{T}}\mathbf{u}/\rho\right)\label{eq:theta}
\end{equation}
in which $\mathbf{z}=(z_{1},z_{2})\triangleq\mathbf{x}-\mathbf{x}_{\text{u}}$,
$\text{sign}(x)=1$ if $x>0$, and $\text{sign}(x)=-1$, otherwise. 

%

%

We now define an alternative expression for the cost function $f(g_{\text{u}}(\mathbf{x}),g_{\text{b}}(\mathbf{x}))$.
\begin{defn}
[Fictitious Segment Cost Function]\label{def:Fictitious-Segment-Cost}
For $\mathbf{x}(\rho,\theta)\in\mathcal{D}_{k}$, $k=1,2,\dots,K$,
the cost function $f(g_{\text{u}}(\mathbf{x}),g_{\text{b}}(\mathbf{x}))$
can be written as 
\[
f(g_{\text{u}}(\mathbf{x}),g_{\text{b}}(\mathbf{x}))=F_{k}(\rho,\theta)
\]
where 
\begin{equation}
F_{k}(\rho,\theta)\triangleq f\big(g_{\text{u}}^{(k)}(\mathbf{x}(\rho,\theta)),g_{\text{b}}(\mathbf{x}(\rho,\theta))\big)\label{eq:Fk}
\end{equation}
and $g_{\text{u}}^{(k)}(\mathbf{x})\triangleq\beta_{k}d_{\text{u}}(\mathbf{x})^{-\alpha_{k}}$
is the UAV-user channel in the $k$th segment from the channel model
(\ref{eq:uav-user-channel-model}). 
\end{defn}

As a result, the objective function in $\mathscr{P}$ is transformed
into the polar domain as $F(\rho,\theta)=\sum_{k=1}^{K}F_{k}(\rho,\theta)\mathbb{I}\{\mathbf{x}(\rho,\theta)\in\mathcal{D}_{k}\}$. 

The motivation of working on the polar domain is that by fixing the
deviation $\theta$, increasing $\rho$ only worsens the propagation
condition according to the nested property (\ref{eq:nested-segmented-model-order-condition})\textendash (\ref{eq:nested-segmented-model-nested-condition}).
In addition, the overall cost function $F(\rho,\theta)$ is discontinuous
because it contains the indicator functions, but the functions $F_{k}(\rho,\theta)$
are continuous. As result, the functions $F_{k}(\rho,\theta)$ can
be used to derive search trajectories.

\subsection{Search Trajectory Design for $K=2$}

\label{subsec:algorithm-K2-case}

The algorithm is better illustrated starting from the two segment
case, where $\mathcal{D}_{1}$ corresponds to the LOS segment and
$\mathcal{D}_{2}$ corresponds to the \ac{nlos} segment. 

\subsubsection{Search on the BS-user Axis}

Let the UAV start from the BS. It first moves towards the user until
it finds two critical positions (if they exist) $\mathbf{x}_{k}^{0}=\mathbf{x}(\rho_{k}^{0},0)$,
$k=1,2$, which correspond to the points achieving the minimum cost
over the BS-user axis in the LOS region and NLOS region, respectively.
Specifically, the parameters $\rho_{k}^{0}$ are the solutions that
minimize the fictitious cost $F_{k}(\rho,0)$ along the BS-user axis
$\mathbf{x}(\rho,0)\in\mathcal{D}_{k}$, for $k=1,2$, in the two
segment case. 

For example, when the UAV is initially in the NLOS region, it can
move up to the LOS-NLOS boundary (see Section \ref{subsec:Propagation-Segment-Detection}
for a discussion on the detection method); with that, it can solve
for $\rho_{1}^{0}$ and $\rho_{2}^{0}$ to obtain the critical points
$\mathbf{x}_{1}^{0}$ and $\mathbf{x}_{2}^{0}$. On the other hand,
when the UAV is initially in the LOS region, it can compute the critical
position $\mathbf{x}_{1}^{0}$ in the LOS region, while $\mathbf{x}_{2}^{0}$
does not exist. 

\subsubsection{Search on the Right Branch}

\label{subsec:search-right-branch}

Starting from the critical position $\mathbf{x}_{1}^{0}$ which minimizes
the cost function on the LOS portion of the BS-user axis, the UAV
first moves to $\mathbf{x}(\rho_{1}^{0},\delta/\rho_{1}^{0})$, \emph{i.e.},
a position just on the right of $\mathbf{x}_{1}^{0}$ in Fig. \ref{fig:UAV-search-path}
(a), where $\delta$ is a chosen step size. It then proceeds according
to the following two phases; at the same time, it keeps the track
record of the minimum cost value $F_{\min}$ discovered and the corresponding
position $\hat{\mathbf{x}}(\hat{\rho},\hat{\theta})$ that achieves
$F_{\min}=F(\hat{\rho},\hat{\theta})$:
\begin{itemize}
\item If the UAV is in the LOS region, it moves away from the user. Specifically,
it moves from $\mathbf{x}(\rho,\theta)$ to $\mathbf{x}(\rho+\delta,\theta)$.
\item If the UAV is in the NLOS region, it moves in the direction that maintains
the same fictitious cost $F_{1}(\rho,\theta)$ as it were in the LOS
region, \emph{i.e.}, contour of $F_{1}(\rho,\theta)=C$ specified
by\footnote{For mathematical completeness, the partial derivative is defined as
$\frac{\partial f(x_{0},y_{0})}{\partial x}=\lim_{t\uparrow0}\frac{1}{t}[f(x_{0}+t,y_{0})-f(x_{0},y_{0})]$
throughout this paper.} 
\[
\frac{\partial F_{1}(\rho,\theta)}{\partial\rho}d\rho+\frac{\partial F_{1}(\rho,\theta)}{\partial\theta}d\theta=0.
\]
It can be shown that (Lemma \ref{lem:partial-optimality} in Appendix
\ref{sec:app-pf-thm-convergence}), $\partial F_{1}(\rho,\theta)/\partial\theta\neq0$
for $\theta\neq0$. As a result from (\ref{eq:polar-representation-xy-plane}),
we have 
\[
d\mathbf{x}=\mathbf{M}(\theta)\mathbf{u}d\rho+\rho\frac{d}{d\theta}\mathbf{M}(\theta)\mathbf{u}\Big(-\frac{\partial F_{1}}{\partial\theta}\Big)^{-1}\frac{\partial F_{1}}{\partial\rho}d\rho.
\]
Thus, the UAV updates its position from $\mathbf{x}$ to $\mathbf{x}+\Delta\mathbf{x}$,
where 
\[
\Delta\mathbf{x}=\gamma\Big[\mathbf{M}(\theta)\mathbf{u}+\rho\frac{d}{d\theta}\mathbf{M}(\theta)\mathbf{u}\Big(-\frac{\partial F_{1}}{\partial\theta}\Big)^{-1}\frac{\partial F_{1}}{\partial\rho}\Big]
\]
in which, $\gamma>0$ is chosen such that the step size $\|\Delta\mathbf{x}\|=\delta$
and the UAV moves in the direction away from the user (see Fig. \ref{fig:UAV-search-path}).
\end{itemize}
The search at this branch is completed whenever the UAV reaches a
point $\mathbf{x}(\rho,\theta)$ such that either $\partial F_{1}(\rho,\theta)/\partial\rho\geq0$
or $\rho\geq L\cos\theta$, where $L\triangleq\|\mathbf{x}_{\text{b}}-\mathbf{x}_{\text{u}}\|$.
The justification of the two stopping criteria will become clear in
Section \ref{subsec:global-optimality}.

\begin{figure}
\begin{centering}
\subfigure[A search path (solid black) shown in the polar domain]{\includegraphics[width=0.8\columnwidth]{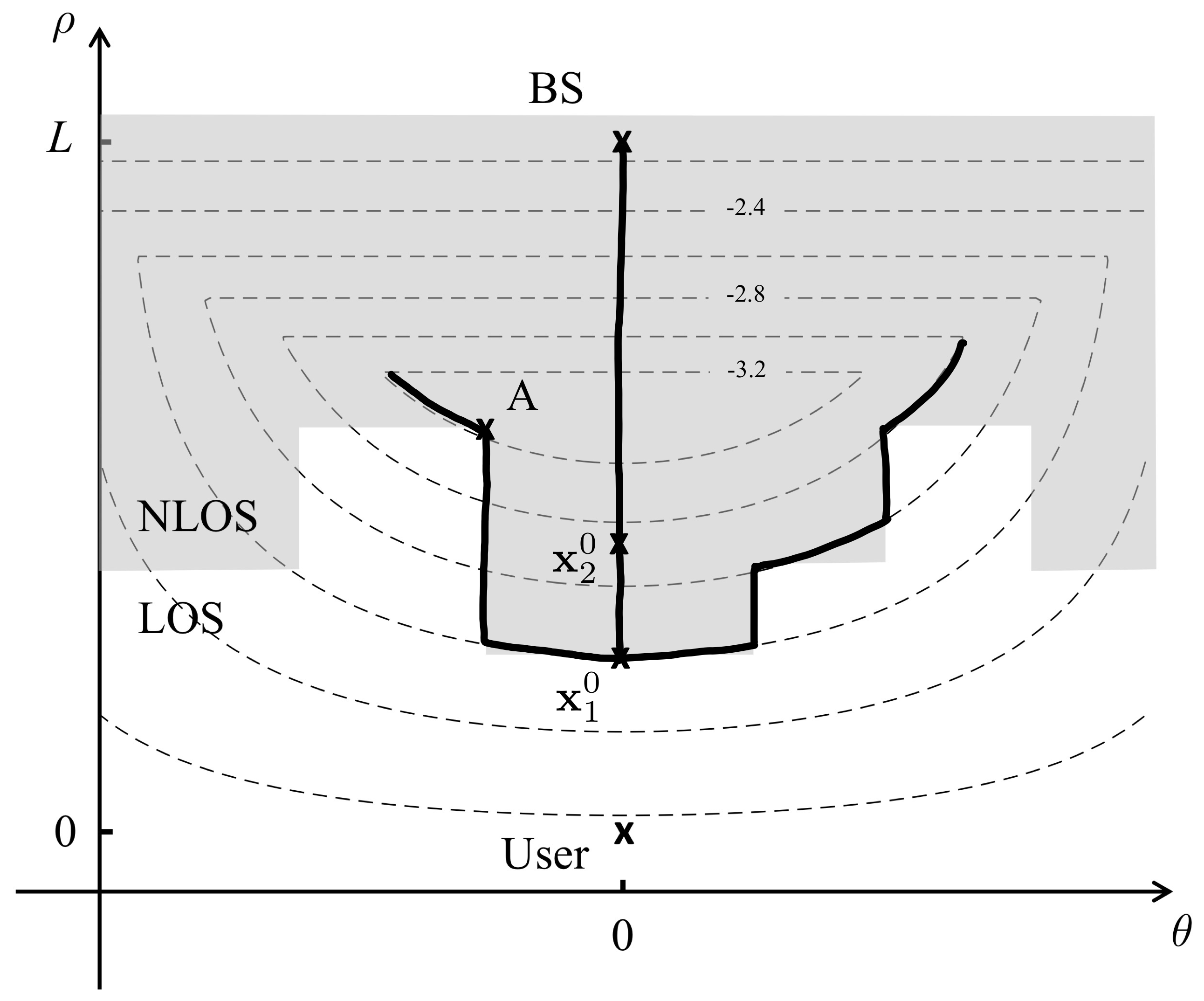}}
\subfigure[A search path (green) shown in the Euclidean domain]{
\includegraphics[width=1\columnwidth]{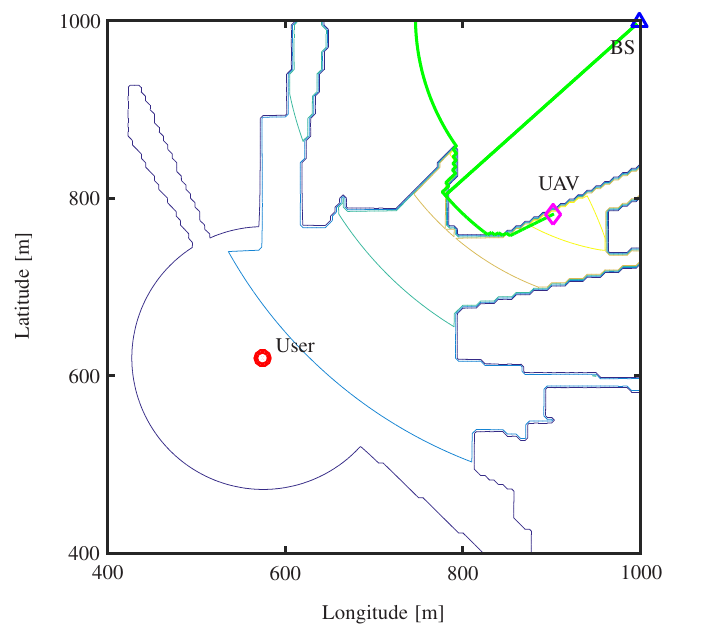}}
\par\end{centering}
\caption{\label{fig:UAV-search-path} (a) An example UAV search trajectory
under two propagation segment case in the polar domain. The dashed
lines represent the contours of the fictitious cost function $F_{1}(\rho,\theta)=C$.
(b) A UAV search path (greed lines) in the Euclidean domain for the
topology in Fig. \ref{fig:urban-city-map} (a). The curves in other
colors represent the contours of the cost function (\ref{eq:drone-rate-maximization}),
and the purple diamond represents the globally optimal position.}
\end{figure}
\begin{algorithm}
Choose a step size $\delta>0$. 
\begin{enumerate}
\item \label{enu:alg-initialization} \textbf{Search on the BS-user axis:}
Find the critical points $\rho_{k}^{0}$, $k=1,2,\dots,K$, defined
in (\ref{eq:problem-critical-point-K}). Initialize $F_{\min}=F_{1}(\rho_{1}^{0},0)$
and $\hat{\mathbf{x}}=\mathbf{x}(\rho_{1}^{0},0)$, where $\mathbf{x}(\rho,\theta)$
is defined in (\ref{eq:polar-representation-xy-plane}). Initialize
$k=1$. 
\item \label{enu:alg-two-phase-update} \textbf{Search on the right branch:}
Set $\mathbf{x}\leftarrow\mathbf{x}(\rho_{k}^{0},\delta/\rho_{k}^{0})$. 
\item \label{enu:alg-two-phase-update-loop}For each iterate $\mathbf{x}$,
compute $\rho=\|\mathbf{x}-\mathbf{x}_{\text{u}}\|$ and $\theta$
from (\ref{eq:theta}).
\begin{enumerate}
\item \textbf{Search in the virtual LOS region:} If $\mathbf{x}\in\bigcup_{j=1}^{k}\mathcal{D}_{j}$,
update 
\begin{equation}
\mathbf{x}\leftarrow\mathbf{x}(\rho+\delta,\theta).\label{eq:alg-eq1}
\end{equation}

When $F(\rho+\delta,\theta)<F_{\min}$, update the record $F_{\min}\leftarrow F(\rho+\delta,\theta)$
and $\hat{\mathbf{x}}\leftarrow\mathbf{x}(\rho+\delta,\theta)$.
\item \textbf{Search in the virtual NLOS region:} If $\mathbf{x}\notin\bigcup_{j=1}^{k}\mathcal{D}_{j}$,
update $\mathbf{x}\leftarrow\mathbf{x}+\Delta\mathbf{x}$, where 
\begin{equation}
\Delta\mathbf{x}=\gamma\bigg[\mathbf{M}(\theta)\mathbf{u}+\rho\frac{d}{d\theta}\mathbf{M}(\theta)\mathbf{u}\Big(-\frac{\partial F_{k}}{\partial\theta}\Big)^{-1}\frac{\partial F_{k}}{\partial\rho}\bigg]\label{eq:alg-eq2}
\end{equation}
where $\gamma>0$ is chosen such that $\|\Delta\mathbf{x}\|=\delta$. 
\end{enumerate}
Repeat this step until either (i) $\rho\geq L\cos\theta$ or (ii)
$\partial F_{k}(\rho,\theta)/\partial\rho\geq0$. 
\item \textbf{\label{enu:alg-two-phase-upadte-loop-left}Search on the left
branch:} Set $\mathbf{x}\leftarrow\mathbf{x}(\rho_{k}^{0},-\delta/\rho_{k}^{0})$.
Repeat Step \ref{enu:alg-two-phase-update-loop}).
\item Let $k\leftarrow k+1$. Repeat from Step \ref{enu:alg-two-phase-update})
until $k>K-1$. 
\item \label{enu:alg-final-step}If $F_{K}(\rho_{K}^{0},0)<F_{\min},$ then
$F_{\min}\leftarrow F_{K}(\rho_{K}^{0},0)$ and $\hat{\mathbf{x}}\leftarrow\mathbf{x}(\rho_{K}^{0},0)$. 
\end{enumerate}
\caption{\label{alg:uav-positioning} Shaded-Contour-Exploration}
\end{algorithm}

\subsubsection{Search on the Left Branch}

\label{subsec:search-left-branch}

Starting from the critical position $\mathbf{x}_{1}^{0}$, the UAV
moves to $\mathbf{x}(\rho_{1}^{0},-\delta/\rho_{1}^{0})$, \emph{i.e.},
a position just on the left of $\mathbf{x}_{1}^{0}$ in Fig. \ref{fig:UAV-search-path}
(a). It repeats the same process as in Section \ref{subsec:search-right-branch}
until it meets the stopping criteria. When the search is completed,
the track record $\hat{\mathbf{x}}$ that achieves the least cost
will be the desired UAV position. Note that the position $\hat{\mathbf{x}}$,
for example, point A in Fig. \ref{fig:UAV-search-path} (a), is \emph{not}
necessarily where the search terminates.

\subsection{Search Trajectory Design for Arbitrary $K$}

\label{subsec:algorithm-search-K}

In the case of more than two propagation segments, one can generate
$K-1$ search trajectories following a similar procedure in Section
\ref{subsec:algorithm-K2-case}, where each trajectory is computed
by partitioning the whole area into two virtual propagation regions:
\emph{virtual LOS} region and \emph{virtual NLOS} region. Specifically,
for the $k$th search trajectory, the virtual LOS region is defined
by grouping the first $k$ propagation segments together $\tilde{\mathcal{D}}_{k}\triangleq\bigcup_{j=1}^{k}\mathcal{D}_{j}$,
and correspondingly, the virtual NLOS region is defined as $\tilde{\mathcal{D}}_{k}^{\text{c}}\triangleq\mathbb{D}\backslash\tilde{\mathcal{D}}_{k}=\bigcup_{j=k+1}^{K}\mathcal{D}_{j}$.
The whole algorithm can be summarized as follows.
\begin{itemize}
\item Search on BS-user axis: Let the UAV start from the BS. It first moves
towards the user along the BS-user axis until it finds the $K-1$
critical positions $\mathbf{x}_{k}^{0}=\mathbf{x}(\rho_{k}^{0},0)$,
$k=1,2,\dots,K-1$, each of which corresponds to a point achieving
the minimum cost over the BS-user axis in a virtual LOS region, \emph{i.e.},
$\rho_{k}^{0}$ are the solutions to 
\begin{align}
\underset{\rho\geq0}{\text{minimize}} & \quad F_{k}(\rho,0)\label{eq:problem-critical-point-K}\\
\text{subject to} & \quad\mathbf{x}(\rho,0)\in\bigcup_{j=1}^{k}\mathcal{D}_{j}\label{eq:problem-critical-point-K-constraint}
\end{align}
for $k=1,2,\dots,K-1$. In addition, $\rho_{K}^{0}$ is defined as
the solution that minimizes (\ref{eq:problem-critical-point-K}) subject
to $\mathbf{x}(\rho,0)\in\mathcal{D}_{K}$. 
\item Search over $K-1$ virtual LOS/NLOS partition scenarios: In the $k$th
search, $k=1,2,\dots,K-1$, the UAV follows a similar procedure as
that in Sections \ref{subsec:search-right-branch} and \ref{subsec:search-left-branch}
for virtual LOS region $\tilde{\mathcal{D}}_{k}$ and virtual NLOS
region $\tilde{\mathcal{D}}_{k}^{\text{c}}$.
\item Integration: During the whole search, the UAV keeps track of the minimum
achievable cost $F_{\min}$ and the corresponding position $\hat{\mathbf{x}}$.
When the algorithm terminates, $\hat{\mathbf{x}}$ gives the desired
UAV position. 
\end{itemize}

The entire search algorithm is summarized in Algorithm \ref{alg:uav-positioning}.
An example search trajectory under the $K=2$ model in the polar coordinate
system is visualized in Fig.~\ref{fig:UAV-search-path} (a), where
the black curve represents the search trajectory and the dashed gray
curves represent the contours of the cost function $F_{1}(\rho,\theta)$.

\subsection{Propagation Segment Detection}

\label{subsec:Propagation-Segment-Detection}

As the parameters $a_{k}$ and $b_{k}$ and the statistics of $\xi_{k}$
in (\ref{eq:uav-user-channle-model-with-noise}) are assumed known
as discussed in Section \ref{subsec:nested-segmented-model}, the
propagation segment can be determined using maximum likelihood detection.
Let $h_{k}$ be the probability distribution function of the random
variable $\xi_{k}$ in propagation segment $k$. Then, the maximum
likelihood estimator of the propagation segment is given by 
\[
\hat{k}=\underset{\scriptsize k=1,2,\dots,K}{\text{argmax}}\;h_{k}(y-b_{k}+a_{k}\log_{10}d(\mathbf{x}))
\]
where $y$ is the channel gain $G_{\text{dB}}$ measured at the UAV
location $\mathbf{x}$. In addition, if $\xi_{k}$ are zero mean Gaussian
distributed with variance $\sigma_{k}^{2}$, then the detection rule
can be simplified as 
\[
\hat{k}=\underset{\scriptsize k=1,2,\dots,K}{\text{argmin}}\;\frac{1}{\sigma_{k}}\left|y-b_{k}+a_{k}\log_{10}d(\mathbf{x})\right|.
\]

Note that, in practice, the parameters $a_{k}$ and $b_{k}$ and the
statistics of $\xi_{k}$ can be estimated in a separate training phase.
Alternatively, the parameter estimation can be integrated in Step
1 in Algorithm \ref{alg:uav-positioning}. Specifically, the UAV moves
along the BS-user axis (and, optionally, to a few random locations)
to collect enough channel measurements. Then, a maximum likelihood
estimation method can be performed to learn these parameters \cite{CheYanGes:C16,CheEsrGesMit:C17}.

\section{Global Optimality and Linear Search Length}

\label{sec:global-convergence}

\textcolor{red}{}

It turns out that Algorithm \ref{alg:uav-positioning} can find the
globally optimal solution to $\mathscr{P}$ for at least two types
of cost functions $f$. 
\begin{condition}
Assume that the cost function $f(x,y)$ in $\mathscr{P}$ satisfies
$\frac{\partial^{2}f(x,y)}{\partial x\partial y}=0$, 
\begin{align}
x\frac{\partial^{2}f(x,y)}{\partial x^{2}}+2\frac{\partial f(x,y)}{\partial x} & \geq0,\quad y\frac{\partial^{2}f(x,y)}{\partial y^{2}}+2\frac{\partial f(x,y)}{\partial y}\geq0\label{eq:cost-function-condition-1}
\end{align}
for every $x,y>0$. 
\end{condition}

It can be easily verified that the cost function (\ref{eq:outage-minimization})
in the outage probability minimization example in Section \ref{subsec:example-AF}
satisfies Condition 1. 
\begin{condition}
Assume that the cost function $f(x,y)$ in $\mathscr{P}$ can be written
as $\max\{f_{1}(x),f_{2}(y)\}$, where $f_{1}(x)$ and $f_{2}(y)$
are decreasing functions. 
\end{condition}

It is also clear that the cost function (\ref{eq:drone-rate-maximization})
in the rate maximization example in Section \ref{subsec:example-DF}
satisfies Condition 2.

In addition, we discuss optimality for continuous-time algorithm trajectory
$\mathbf{x}(t)$, which can be obtained from Algorithm \ref{alg:uav-positioning}
using infinitesimal step size $\delta=\mathcal{O}(dt)$ at each infinitesimal
time slot $dt$. Specifically, the search trajectory $\mathbf{x}(t)$
in Algorithm \ref{alg:uav-positioning} can be described by piece-wise
continuous dynamic systems, where one replaces $\delta$ by $\kappa dt$
in (\ref{eq:alg-eq1}) and $\gamma$ by $\kappa\bar{\gamma}dt$ in
(\ref{eq:alg-eq2}), in which $\kappa$ is a parameter that specifies
the moving speed of the UAV. Accordingly, the continuous-time processes
of the minimum cost $F_{\min}(t)$ and the position track record $\hat{\mathbf{x}}(t)$
are given by $F_{\min}(t)=\text{minimize}_{0\leq\tau\leq t}\;f(g_{\text{u}}(\mathbf{x}(\tau),g_{\text{b}}(\mathbf{x}(\tau))$
and $\hat{\mathbf{x}}(t)=\mathbf{x}(\hat{\tau})$, respectively, where
$\hat{\tau}=\arg\min_{0\leq\tau\leq t}\;f(g_{\text{u}}(\mathbf{x}(\tau),g_{\text{b}}(\mathbf{x}(\tau))$. 

\subsection{Global Optimality}

\label{subsec:global-optimality}

We first present the main optimality result as follows. \textcolor{red}{}
\begin{thm}
[Global Optimality]\label{thm:convergence-two-segment} Suppose
that the cost function $f$ in $\mathscr{P}$ satisfies either Condition
1 or Condition 2. Then, $\hat{\mathbf{x}}(t)$ in Algorithm \ref{alg:uav-positioning}
converges to the globally optimal solution to $\mathscr{P}$ and $F_{\min}(t)$
converges to the minimum cost value in finite time. 
\end{thm}

Theorem \ref{thm:convergence-two-segment} confirms that the globally
optimal UAV position is attainable, even though the terrain topology
could be arbitrarily complex. 

The optimality result can be better understood from the polar coordinate
system. From the definition of the fictitious segment cost functions
$F_{k}(\rho,\theta)$ in (\ref{eq:Fk}), problem $\mathscr{P}$ can
be equivalently written as 
\[
\mathscr{P}':\quad\underset{\rho\geq0,-\pi\leq\theta\leq\pi}{\text{minimize}}\quad F(\rho,\theta)\triangleq\sum_{k=1}^{K}F_{k}(\rho,\theta)\mathbb{I}\{(\rho,\theta)\in\mathcal{P}_{k}\}
\]
where $\mathcal{P}_{k}\triangleq\big\{(\rho,\theta):\mathbf{x}(\rho,\theta)\in\mathcal{D}_{k}\big\}$
is the $k$th propagation segment in the polar coordinate system.
The optimal solution $\mathbf{x}^{\star}$ to $\mathscr{P}$ can be
obtained as $\mathbf{x}^{\star}=\mathbf{x}(\rho^{\star},\theta^{\star})$,
in which $(\rho^{\star},\theta^{\star})$ is the optimal solution
to $\mathscr{P}'$. 

The following intermediate results provide some intuitions to understand
Algorithm \ref{alg:uav-positioning} and Theorem \ref{thm:convergence-two-segment}. 
\begin{prop}
[Bounded Search Region]\label{prop:Search-region} The optimal solution
$\mathbf{x}^{\star}$ to $\mathscr{P}$ can be obtained as $\mathbf{x}(\rho^{\star},\theta^{\star})$,
where $(\rho^{\star},\theta^{\star})\in\mathcal{P}$ and 
\begin{equation}
\mathcal{P}=\Big\{(\rho,\theta):0\leq\rho\leq L\cos\theta,-\frac{\pi}{2}\leq\theta\leq\frac{\pi}{2}\Big\}\label{eq:search-region}
\end{equation}
in which, $L\triangleq\|\mathbf{x}_{\text{b}}-\mathbf{x}_{\text{u}}\|$
is the horizontal distance from the BS to the user. 
\end{prop}
\begin{proof}
Please refer to Appendix \ref{app:pf-prop-search-region}. 
\end{proof}

Proposition \ref{prop:Search-region} justifies the first stopping
criterion $\rho\geq L\cos\theta$ in Step \ref{enu:alg-two-phase-update-loop})
of Algorithm \ref{alg:uav-positioning}. An intuitive explanation
is that when the UAV moves outside the region $\mathcal{P}$ in (\ref{eq:search-region}),
one can always find a position in $\mathcal{P}$ that has an equal
(or smaller) distances, respectively, to the BS and to the user in
the same (or less obstructed) propagation segment $\mathcal{D}_{k}$,
leading to an equal (or lower) cost to achieve. Therefore, the optimal
UAV position is contained in $\mathcal{P}$. 

The following proposition justifies the second stopping criterion
$\partial F_{k}(\rho,\theta)/\partial\rho\geq0$. 
\begin{prop}
[Partial Optimality]\label{prop:Partial-Optimality} Suppose that
the cost function $f$ in $\mathscr{P}$ satisfies either Condition
1 or Condition 2. Then, $F_{k}(\rho,\theta)$ admits a unique local
minimizer $\rho_{k}^{*}(\theta)$ over $\rho\geq0$ for every fixed
$\theta$, where $|\theta|<\pi/2$. 
\end{prop}
\begin{proof}
Please refer to Appendix \ref{app:pf-prop-partial-optimality}.
\end{proof}
Due to the fact that $F_{k}(\rho,\theta)$ has a unique local minimum
for every $\theta$, condition $\partial F_{k}(\rho,\theta)/\partial\rho\geq0$
implies that $\rho\geq\rho_{k}^{*}(\theta)$. On the other hand, Step
\ref{enu:alg-two-phase-update-loop}) of Algorithm \ref{alg:uav-positioning}
always increases $\rho$ (to be formally justified in Lemma \ref{lem:monotonicity}
in Appendix \ref{app:pf-prop-maximum-traj-length}). Therefore, it
suffices to stop the search when the condition $\partial F_{k}(\rho,\theta)/\partial\rho\geq0$
is met.

Based on Propositions \ref{prop:Search-region} and \ref{prop:Partial-Optimality},
the proof of Theorem \ref{thm:convergence-two-segment} is derived
in Appendix \ref{sec:app-pf-thm-convergence}.

\subsection{Maximum Length of the Algorithm Trajectory }

Here, we derive the worst-case trajectory length of Algorithm \ref{alg:uav-positioning}. 
\begin{thm}
[Maximum Trajectory Length]\label{thm:Maximum-Trajectory-Length}
The length of the search trajectory from Algorithm \ref{alg:uav-positioning}
is upper bounded by $(2.4K-1.4)L$, where $L\triangleq\|\mathbf{x}_{\text{b}}-\mathbf{x}_{\text{u}}\|$
is the horizontal distance from the BS to the user. 
\end{thm}
\begin{proof}
Please refer to Appendix \ref{app:pf-prop-maximum-traj-length}.
\end{proof}
%


Theorem \ref{thm:Maximum-Trajectory-Length} suggests that the algorithm
must terminate in a finite number of steps given a positive step size
$\delta>0$. The total number of steps scales as $\mathcal{O}(L/\delta)$.
\emph{Surprisingly}, the bound is linear in $L$ and does \emph{not}
depend on the actual terrain, \emph{i.e.}, the shapes of the propagation
segments $\mathcal{D}_{k}$. As a benchmark, if one searches the optimal
UAV position following the segment boundaries (a property from Proposition
\ref{prop:Solution-Property}), the worst-case search length is not
guaranteed to be linear in $L$, depending on the actual shapes of
the boundaries (see, for example, Fig. \ref{fig:urban-city-map}(c)). 

\textcolor{red}{}

\section{Numerical Results}

\label{sec:numerical}

Consider a dense urban area with buildings ranging from $5$--$45$
meter height following a uniform distribution as illustrated in Fig.
\ref{fig:urban-city-map} (a). The user is represented by a red circle
and the BS locates at the top right corner denoted by a blue triangle.
The height of the BS is $45$ meters, and the UAV moves at $50$ meter
above the ground. As a result, there is always LOS propagation between
the UAV and the BS. Consider two propagation scenarios, LOS and NLOS,
for the UAV-to-user link; this corresponds to many existing models
in the literature for a fair comparison. Correspondingly, the parameters
of the UAV-BS channel in (\ref{eq:UAV-BS-channel-model}) are chosen
as $(\alpha_{0},\log_{10}\beta_{0})=(2.08,\,-3.85)$; Rician fading
with 20 dB K-factor is assumed according to the Rural Macro BS to
UAV scenario in \cite{TR36777}. The parameters of the UAV-user channel
in (\ref{eq:uav-user-channel-model}) are chosen as $(\alpha_{1},\log_{10}\beta_{1},\alpha_{2},\log_{10}\beta_{2})=(2.14,\,-3.69,\,3.03,\,-3.84)$;
Rician fading with 9 dB K-factor is assumed for the LOS case and Rayleigh
fading is assumed for the NLOS case according to the Urban Micro BS
to UAV scenario in \cite{TR36777}. 

We consider the example capacity maximization problem in Section \ref{subsec:example-DF}
for optimal UAV positioning. The transmission powers are chosen as
$P_{\text{b}}=30$ dBm from the BS and $P_{\text{d}}=36$ dBm from
the UAV, and the noise power is $-80$ dBm. The corresponding power
map and end-to-end capacity map for every UAV position are illustrated
in Fig. \ref{fig:urban-city-map} (b) and (c). 

In Fig. \ref{fig:UAV-search-path} (b), the green curves show the
search path from Algorithm \ref{alg:uav-positioning}. The two curve
branches correspond to the UAV searches in Step \ref{enu:alg-two-phase-update})
and Step \ref{enu:alg-two-phase-upadte-loop-left}) in Algorithm \ref{alg:uav-positioning},
respectively. The other curves in Fig. \ref{fig:UAV-search-path}
(b) represent the contour of the capacity map as in Fig. \ref{fig:urban-city-map}
(c). The optimal UAV position is found at the purple diamond. 

\subsection{Throughput and Outage Probability for Single User Case}

\begin{figure}
\begin{centering}
\includegraphics[width=1\columnwidth]{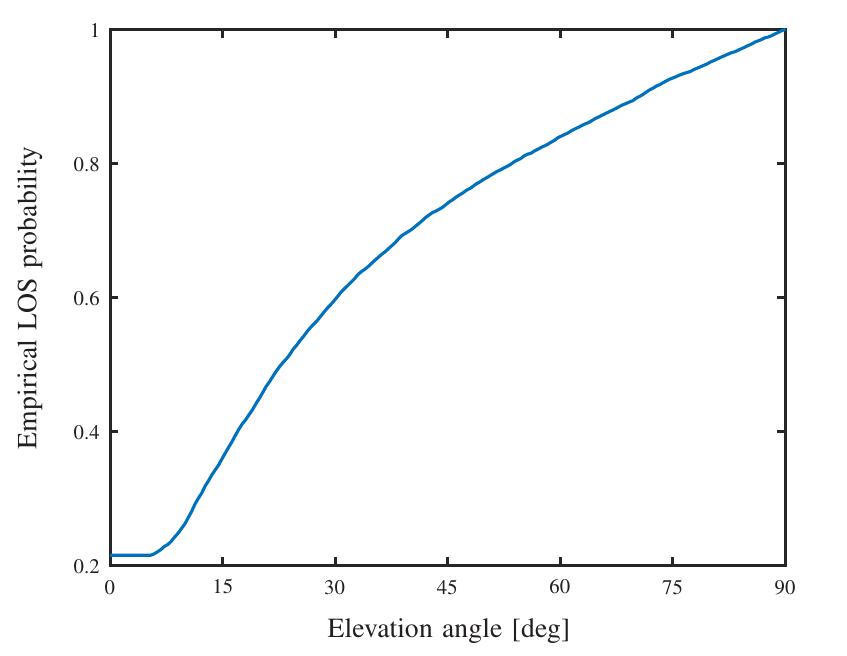}
\par\end{centering}
\caption{\label{fig:plos} Empirical cumulative distribution function of the
LOS probability in terms of the elevation angle estimated for the
city topology shown in Fig. \ref{fig:urban-city-map} (a). This is
for evaluating the baseline scheme of Probabilistic Algorithm. }
\end{figure}
We now evaluate the average throughput of UAV relay system for a single
user located randomly and uniformly on the streets in Fig. \ref{fig:urban-city-map}
(a). The transmission powers are chosen as $P_{\text{b}}=P_{\text{d}}=33$
dBm for both the BS and the UAV, and the throughput is evaluated as
the absolute value of (\ref{eq:drone-rate-maximization}). The proposed
scheme places the UAV to the position obtained from Algorithm \ref{alg:uav-positioning}
under step size $\delta=5$ meters. We also consider the following
baselines for UAV positioning: 
\begin{itemize}
\item \textbf{Probabilistic Algorithm }\cite{MozSaaBenDeb:C16}\textbf{:}
First, obtain an empirical LOS distribution $f_{\text{LOS}}(\varphi)=\mathbb{P}\{\text{LOS}\,,\varphi\}$
as a function of the elevation angle $\varphi=\varphi(\mathbf{x},\mathbf{x}_{\text{u}})$
from the user at $(\mathbf{\mathbf{x}_{\text{u}}},0)$ to the UAV
at $(\mathbf{x},H_{\text{d}})$, by uniformly and randomly dropping
10,000 users on the streets in Fig. \ref{fig:urban-city-map} (a)
and randomly picking UAV locations within the target area. The empirical
distribution function $f_{\text{LOS}}(\varphi)$ is shown in Fig.
\ref{fig:plos}. Second, given each user position $\mathbf{x}_{\text{u}}$,
the UAV-user channel at UAV location $\mathbf{x}$ is computed as
\begin{align}
g_{\text{u}}(\mathbf{x},\mathbf{x}_{\text{u}}) & =f_{\text{LOS}}(\varphi)\beta_{1}\|\mathbf{x}-\mathbf{x}_{\text{u}}\|^{-\alpha_{1}}\nonumber \\
 & \qquad+\big(1-f_{\text{LOS}}(\varphi)\big)\beta_{2}\|\mathbf{x}-\mathbf{x}_{\text{u}}\|^{-\alpha_{2}}.\label{eq:offline-gu}
\end{align}
The optimal UAV position is then obtained by solving $\mathscr{P}$
with $g_{\text{u}}(\mathbf{x)=}g_{\text{u}}(\mathbf{x},\mathbf{x}_{\text{u}})$
in (\ref{eq:offline-gu}). 
\item \textbf{Simple Search:} Obtain the optimal UAV position by searching
only on the BS-user axis (\emph{i.e.}, to implement only Step 1 of
Algorithm \ref{alg:uav-positioning}).
\item \textbf{Exhaustive Search:} We perform exhaustive search over the
entire search region on equally-spaced grids with $\delta=5$ meter
spacing. The grid point that maximizes the cost in (\ref{eq:drone-rate-maximization})
is chosen as the UAV position. Note that such a scheme is prohibited
in practice and hence it is for benchmarking only. 
\end{itemize}
The performance on direct BS-user transmission (without UAV relaying)
is evaluated using the segmented channel model (\ref{eq:uav-user-channel-model})
by replacing $d_{\text{u}}(\mathbf{x})$ by the BS-user distance and
replacing $\mathbb{I}\{\mathbf{x}\in\mathcal{D}_{k}\}$ by the indicator
of the BS-user link propagation condition. 

\begin{figure}
\begin{centering}
\subfigure[Throughput of the three user categories]{
\includegraphics[width=1\columnwidth]{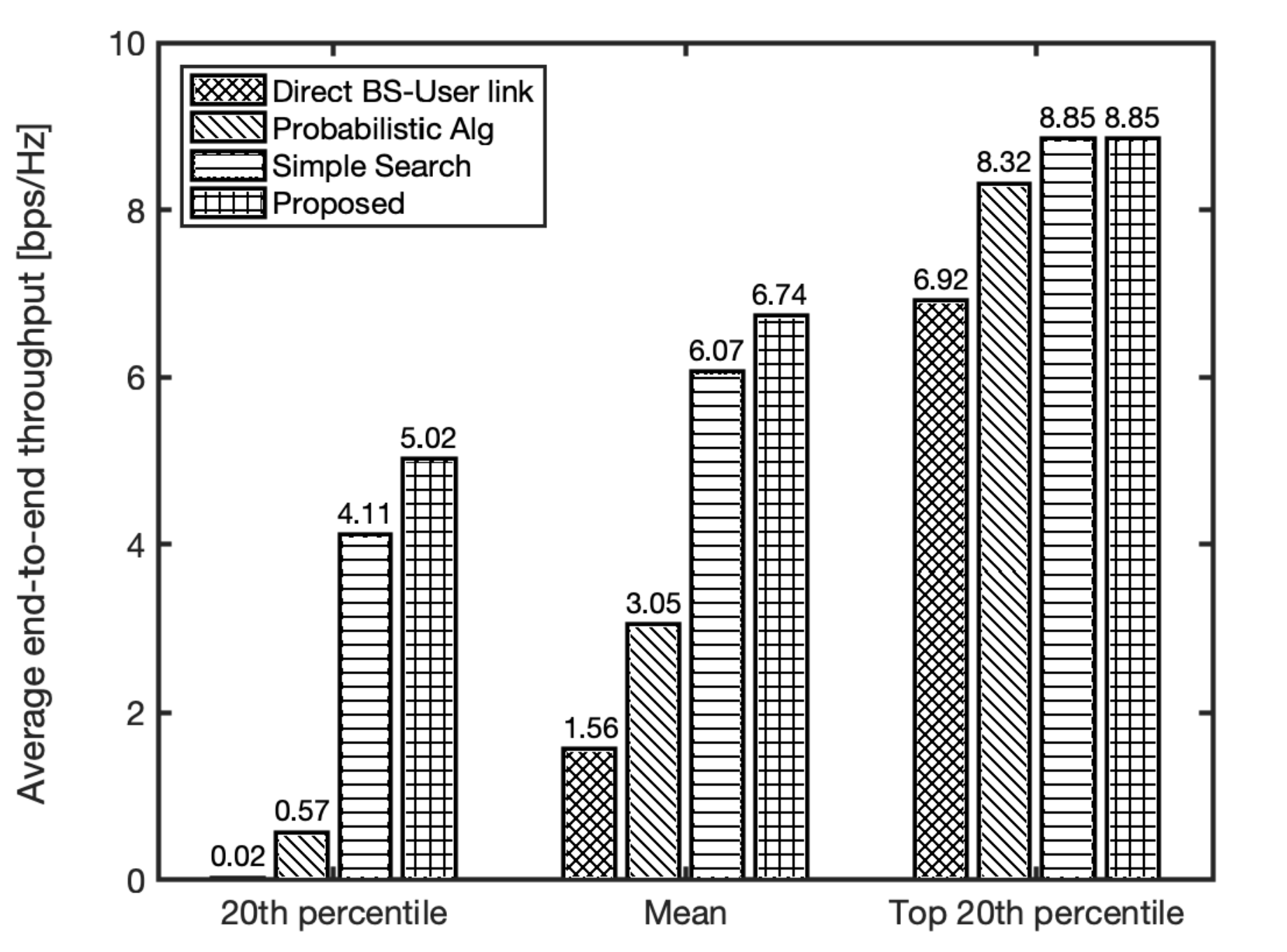}}

\subfigure[CDF of the end-to-end throughput]{\includegraphics[width=1\columnwidth]{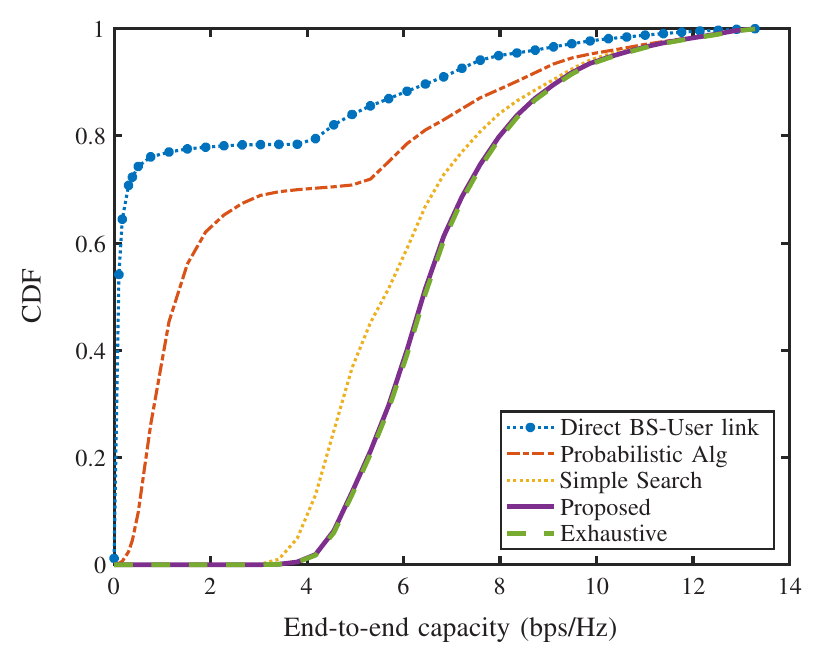}}
\par\end{centering}
\caption{\label{fig:rate-bar} Comparison of the average end-to-end throughput
over different schemes. }
\end{figure}
Fig. \ref{fig:rate-bar} compares the average capacity of the Decode-and-Forward
relay system discussed in Section \ref{subsec:example-DF} over 10,000
user positions in Fig. \ref{fig:urban-city-map} (a). The cell edge
users are recognized as those within the 20th percentile of the throughput
under direct BS-user transmission, and the cell center users are recognized
as those in the top 20th percentile.\footnote{We found from our experiment that around 22\% users have LOS condition
on the BS-user link.} First, across all the three user categories shown in Fig. \ref{fig:rate-bar}
(a), the proposed scheme with the optimal UAV placement achieves the
highest throughput. In particular, it realizes more than 2X throughput
gain on average over an probabilistic approach. Second, the proposed
scheme performs as well as the exhaustive search scheme as seen from
the two overlapping \ac{cdf} curves in Fig. \ref{fig:rate-bar} (b),
which verify our theoretical results on optimality. 

Fig. \ref{fig:outage} compares the outage probability reduction from
the direct BS-user transmission scheme using the Amplify-and-Forward
relay system discussed in Section \ref{subsec:example-AF}. Specifically,
it is defined as $f(\mathbf{x}^{*})/f_{0}$, where $f(\mathbf{x})$
is the cost function defined in (\ref{eq:outage-minimization}) with
$\mathbf{x}^{*}$ being the optimal UAV positions found by various
schemes, and $f_{0}=1/(P_{\text{b}}g_{0}(\mathbf{x}_{\text{u}}))$,
in which, $g_{0}(\mathbf{x}_{\text{u}})$ is the channel gain of the
direct BS-user link. It is observed that the proposed scheme provides
a significant outage probably reduction. In addition, it also confirms
the globally optimality of the proposed scheme, when comparing its
performance with the exhaustive search baseline. 

\begin{figure}
\begin{centering}
\includegraphics[width=1\columnwidth]{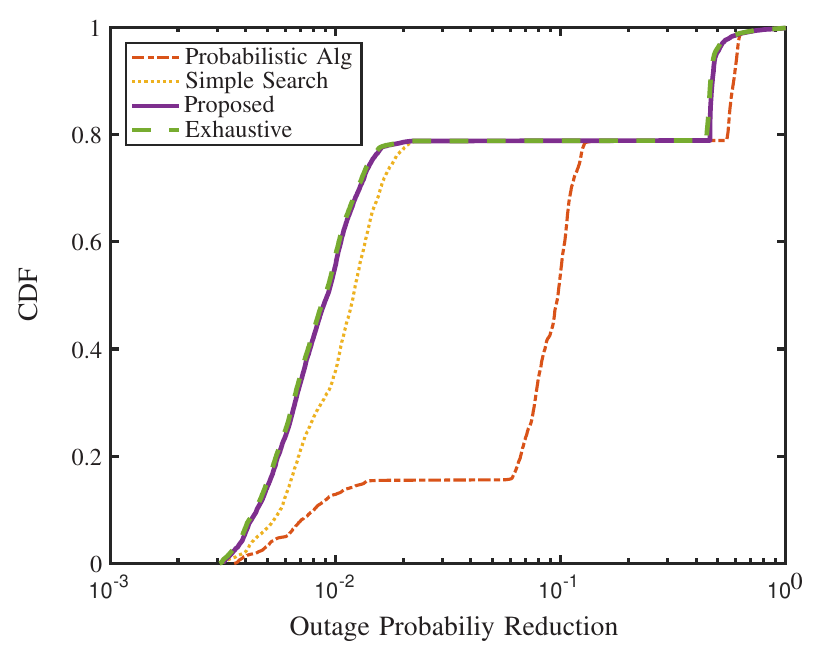}
\par\end{centering}
\caption{\label{fig:outage} Outage probability reduction from the direct BS-user
transmission scheme.}
\end{figure}
\begin{figure}
\begin{centering}
\includegraphics[width=1\columnwidth]{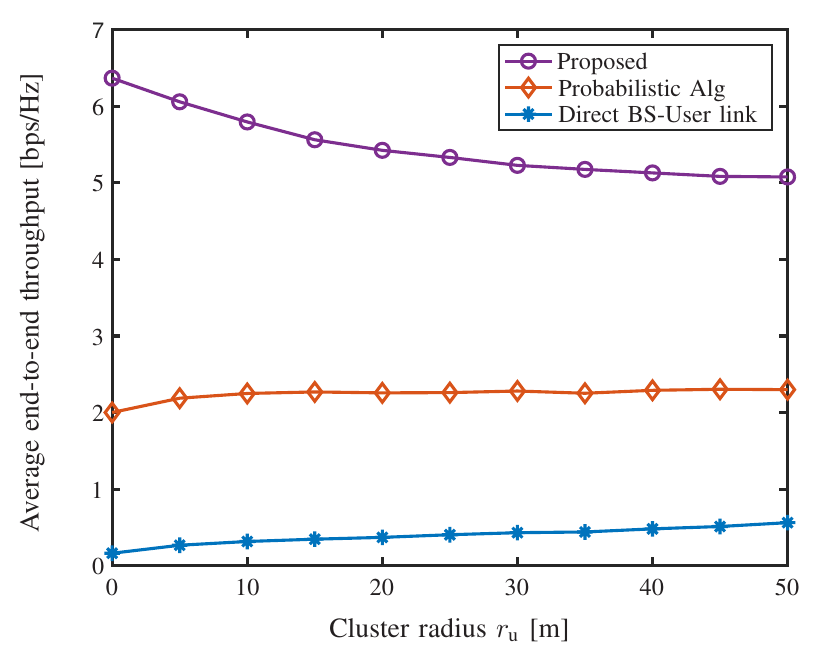}
\par\end{centering}
\caption{\label{fig:rate-radius-1} Average throughput for a cluster of $N_{\text{u}}=20$
users distributed with radius at most $r_{\text{u}}$ from the center
of a hotspot.}
\end{figure}

\subsection{Clustered Multiuser Case}

\label{subsec:clustered-user-case}

We evaluate the proposed UAV placement strategy for the case of $N_{\text{u}}=20$
users clustered around a hotspot with radius $r_{\text{u}}$ meters.
We implement a search trajectory from Algorithm \ref{alg:uav-positioning}
\ac{wrt} a virtual user located at $\mathbf{x}_{\text{c}}$ and evaluate
the sum data rate $\bar{f}(\mathbf{x})$ as discussed in Section \ref{subsec:multiuser}.
As a benchmark, the baseline scheme minimizes the average cost $\bar{f}(\mathbf{x})$
evaluated using the probabilistic model $g_{\text{u}}(\mathbf{x},\mathbf{x}_{\text{u}}^{(i)})$
in (\ref{eq:offline-gu}) for each user $i$. We drop 500 clusters
uniformly and randomly in Fig. \ref{fig:urban-city-map} (a), and
evaluate the average throughput for the \emph{obstructed} clusters,
of which the center position $\mathbf{x}_{\text{c}}$ is in \ac{nlos}
condition to the BS.

Fig. \ref{fig:rate-radius-1} shows the average throughput performance
versus the radius $r_{\text{u}}$ of the cluster. The case $r_{\text{u}}=0$
refers to the single user scenario, and we know that Algorithm \ref{alg:uav-positioning}
achieves the globally optimal performance. When the cluster radius
increases, the average throughput drops, because the proposed search
algorithm is not optimized for the multiple users away from the cluster
center. Nevertheless, the system still benefits from the proposed
search strategy because the users nearby still share some degree of
correlation in the propagation environment. As a result, the proposed
strategy does also achieve significant performance gain over the probabilistic
scheme. 

\section{Conclusion}

\label{sec:conclusion} 

This paper designed algorithms to search for the optimal UAV position
for establishing the best wireless relay link between a BS and a user
in a dense urban area, where global topological information, such
as a 3D city map, is not available. The desired UAV position, constrained
at a fixed height above the ground, was defined as the minima of a
general cost function. A nested segmented propagation model was proposed
to model the propagation from the UAV to the ground user that is probably
blocked by obstacles. A search algorithm was developed and shown to
find the globally optimal UAV position. In addition, the length of
the search trajectory is upper bounded by a linear function of the
diameter of the target area. Significant throughput gain was found
when compared to other UAV positioning approaches or direct BS-user
transmissions.

\appendices

\section{Proof of Proposition \ref{prop:Search-region}}

\label{app:pf-prop-search-region}

From the nested segmented propagation model (\ref{eq:nested-segmented-model-order-condition})\textendash (\ref{eq:nested-segmented-model-nested-condition}),
it is easy to verify that for any $\rho$ and $\theta$, the following
holds 
\begin{equation}
F_{k}(\rho,\theta)<F_{k+1}(\rho,\theta)\label{eq:inequality-less-obstructed-segment-preferred}
\end{equation}
for $k=1,2,\dots,K-1$.

The result of Proposition \ref{prop:Search-region} can be proven
by contradiction as follows. 

First, suppose that solution $\mathbf{x}(\rho',\theta')$ satisfies
$\frac{\pi}{2}<|\theta'|<\pi$. Then, we must have $F_{1}(\rho',\theta')\leq F(\rho',\theta')$
from property (\ref{eq:inequality-less-obstructed-segment-preferred})
as $\mathcal{P}_{k}$ are mutually exclusive. Since there is LOS propagation
when the UAV is on top of the outdoor user, \emph{i.e.}, $(0,0)\in\mathcal{P}_{1}$,
we have 
\[
F(0,0)=F_{1}(0,0)<F_{1}(\rho',\theta')\leq F(\rho',\theta')
\]
where, by moving from $(0,0)$ to $(\rho',\theta')$, the UAV position
$(\rho',\theta')$ has longer distances to both the user and the BS.
Thus, by contradiction, $\mathbf{x}(\rho',\theta')$ cannot be the
optimal solution to $\mathscr{P}$, \emph{i.e.}, we must have $|\theta'|\leq\frac{\pi}{2}$.

Second, suppose that the optimal solution $\mathbf{x}(\rho',\theta')$
satisfies $\rho'>L\cos\theta'$. Consider a different solution $\mathbf{x}(\rho'',\theta')$,
where $\rho''=L\cos\theta'-(\rho'-L\cos\theta')$. From geometry,
one can easily show that the UAV-BS distance $d_{\text{b}}(\mathbf{x}(\rho',\theta'))=d_{\text{b}}(\mathbf{x}(\rho'',\theta'))$
and the UAV-user distance $d_{\text{u}}(\mathbf{x}(\rho'',\theta'))<d_{\text{u}}(\mathbf{x}(\rho',\theta'))$.
Suppose that $(\rho',\theta')\in\mathcal{P}_{k}$ and $(\rho'',\theta')\in\mathcal{P}_{j}$.
From the nested segmented property (\ref{eq:nested-segmented-model-order-condition})\textendash (\ref{eq:nested-segmented-model-nested-condition}),
we must have $j\leq k$. From property (\ref{eq:inequality-less-obstructed-segment-preferred}),
we have
\[
F(\rho'',\theta')=F_{j}(\rho'',\theta')<F_{k}(\rho'',\theta')<F_{k}(\rho',\theta')=F(\rho',\theta')
\]
which contradicts to the hypothesis that $(\rho',\theta')$ minimizes
the objective function $F(\rho,\theta)$. This confirms that the optimal
solution must satisfies $0\leq\rho\leq L\cos\theta$.

\section{Proof of Proposition \ref{prop:Partial-Optimality}}

\label{app:pf-prop-partial-optimality}

\subsection{Proof under Condition I}

We first analyze the property of the composite function $f((g_{\text{u}}\circ d_{\text{u}})(\mathbf{z}),(g_{\text{b}}\circ d_{\text{b}})(\mathbf{z}))$. 
\begin{lem}
\label{lem:for-condition-1}Let $g_{i}(d_{i})=\beta_{i}d_{i}^{-\alpha_{i}}$,
$i=1,2$, and $\alpha>1$. Let $d_{i}(\mathbf{z})$, $i=1,2$, be
convex functions in $\mathbf{z}$. Suppose that $f(x,y)$ satisfies
Condition 1. Then, the composite function $f((g_{1}\circ d_{1})(\mathbf{z}),(g_{2}\circ d_{2})(\mathbf{z}))$
is strictly convex in $\mathbf{z}$. 
\end{lem}
\begin{proof}
Since $\partial f/\partial x<0$ and $\partial f/\partial y<0$, the
first order partial derivative of $f(g_{\text{1}}\circ d_{1},g_{2}\circ d_{2})$
is given by 
\begin{align*}
\frac{\partial f}{\partial d_{i}} & =\frac{\partial f}{\partial g_{i}}\frac{\partial g_{i}}{\partial d_{i}}=\frac{\partial f}{\partial g_{i}}(-\alpha_{i}\beta_{i})d_{i}^{-\alpha_{i}-1}>0
\end{align*}
for $i=1,2$, and the second order partial derivative is given by
\begin{align*}
\frac{\partial^{2}f}{\partial d_{i}^{2}} & =\frac{\partial^{2}f}{\partial g_{i}^{2}}\frac{\partial g_{i}}{\partial d_{i}}(-\alpha_{i}\beta_{i})d_{i}^{-\alpha_{i}-1}+(-\alpha_{i}\beta_{i})\frac{\partial f}{\partial g_{i}}\frac{\partial}{\partial d_{i}}d_{i}^{-\alpha_{i}-1}\\
 & =\alpha_{i}^{2}\beta_{i}d_{i}^{-\alpha_{i}-2}\Big(\beta_{i}d_{i}^{-\alpha_{i}}\frac{\partial^{2}f}{\partial g_{i}^{2}}+\frac{\alpha_{i}+1}{\alpha_{i}}\frac{\partial f}{\partial g_{i}}\Big)\\
 & >\alpha_{i}^{2}\beta_{i}d_{i}^{-\alpha_{i}-2}\Big(g_{i}\frac{\partial^{2}f}{\partial g_{i}^{2}}+2\frac{\partial f}{\partial g_{i}}\Big)\\
 & \geq0
\end{align*}
for $i=1,2$, where the first inequality is due to the fact that $\alpha_{i}>1$
and $\frac{\partial f}{\partial g_{i}}<0$. 

Define an operator $\nabla\triangleq[\frac{\partial}{\partial z_{1}}\,\frac{\partial}{\partial z_{2}}\,\dots\,\frac{\partial}{\partial z_{m}}]^{\text{T}}$,
where $z_{i}$ is the $i$th entry of a vector variable $\mathbf{z}$.
From $\nabla f=\frac{\partial f}{\partial d_{1}}\nabla d_{1}+\frac{\partial f}{\partial d_{2}}\nabla d_{2}$,
the Hessian matrix of $f$ is given by 
\begin{align*}
\nabla^{2}f & =\frac{\partial^{2}f}{\partial d_{1}^{2}}\nabla d_{1}\nabla d_{1}^{\text{T}}+\frac{\partial f}{\partial d_{1}}\nabla^{2}d_{1}\\
 & \qquad+\frac{\partial^{2}f}{\partial d_{2}^{2}}\nabla d_{2}\nabla d_{2}^{\text{T}}+\frac{\partial f}{\partial d_{2}}\nabla^{2}d_{2}\succeq\mathbf{0}
\end{align*}
since $\nabla^{2}d_{i}(\mathbf{z})\succeq\mathbf{0}$ due to the convexity
of $d_{i}(\mathbf{z})$. Therefore, $f$ is strictly convex in $\mathbf{z}$
(as $d_{i}(\mathbf{z})$ are strictly convex). 
\end{proof}

From the polar representation of the UAV-user and UAV-BS distances
\begin{align}
d_{\text{b}}(\mathbf{x}(\rho,\theta)) & =\sqrt{\rho^{2}+L^{2}-2\rho L\cos\theta+(H_{\text{d}}-H_{\text{b}})^{2}}\label{eq:dist_b}\\
d_{\text{u}}(\mathbf{x}(\rho,\theta)) & =\sqrt{\rho^{2}+(H_{\text{d}}-H_{\text{u}})^{2}}\label{eq:dist_u}
\end{align}
one can show that $d_{\text{b}}(\mathbf{x}(\rho,\theta))$ and $d_{\text{u}}(\mathbf{x}(\rho,\theta))$
are strictly convex in $\rho$.

Then, using Lemma \ref{lem:for-condition-1} and the definition of
$F_{k}(\rho,\theta)$ in (\ref{eq:Fk}), we can conclude that $F_{k}(\rho,\theta)$
is strictly convex in $\rho$, and therefore, $F_{k}(\rho,\theta)$
admits a unique local minima $\rho_{k}^{*}(\theta)$ in the bounded
interval $\rho\in[0,L\cos\theta]$.

\subsection{Proof under Condition II}

Consider $F_{k}(\rho,\theta)=\max\{f_{1}(g_{\text{u}}^{(k)}(\mathbf{x}(\rho,\theta)),f_{2}(g_{\text{b}}(\mathbf{x}(\rho,\theta))\}$.
From (\ref{eq:dist_b})\textendash (\ref{eq:dist_u}), to increase
$\rho$, we must have $d_{\text{u}}$ increase and $d_{\text{b}}$
decrease, and as a result, $f_{1}$ decreases and $f_{2}$ increases
monotonically in $\rho\in[0,L\cos\theta]$. Therefore, these exists
a unique local minimizer $\rho_{k}^{*}(\theta)$ in the closed interval
$[0,L\cos\theta]$. 

%


%

%


\section{Proof of Theorem \ref{thm:convergence-two-segment} }

\label{sec:app-pf-thm-convergence}

We first state the following lemma for the property of the fictitious
segment cost functions $F_{k}(\rho,\theta)$ in (\ref{eq:Fk}) derived
in the polar domain from the original objective $f$. 
\begin{lem}
\label{lem:partial-optimality} It holds that 
\begin{equation}
\partial F_{k}(\rho,\theta)/\partial|\theta|>0\label{eq:local-optimality-condition-0}
\end{equation}
for all $k$ and $\theta\neq0$. In addition, for any $\theta'\geq0$,
the following property holds 
\begin{equation}
\min_{\rho\geq0}F_{k}(\rho,\theta')\quad\leq\quad\min_{j\geq k}\;\min_{\rho\geq0,\theta'<\theta\leq\frac{\pi}{2}}F_{j}(\rho,\theta)\label{eq:local-optimality-condition-1}
\end{equation}
for $1\leq k\leq K$. Similarly, for any $\theta'<0$, 
\begin{equation}
\min_{\rho\geq0}F_{k}(\rho,\theta')\quad\leq\quad\min_{j\geq k}\;\min_{\rho\geq0,-\frac{\pi}{2}\leq\theta<\theta'}F_{j}(\rho,\theta)\label{eq:local-optimality-condition-2}
\end{equation}
\end{lem}
\begin{proof}
We first note that $\frac{\partial}{\partial|\theta|}F_{k}(\rho,\theta)>0$
for all $k=1,2,\dots,K$, because increasing $|\theta|$ will increase
the UAV-BS distance $d_{\text{b}}$ while the UAV-user distance $d_{\text{u}}=\rho$
is not affected, and hence $g_{\text{b}}(\mathbf{x}(\rho,\theta))$
is decreased. As the cost function $f(x,y)$ is increasing with $x$
and $y$, respectively, due to conditions I or II, the cost $F_{k}(\rho,\theta)$
increases as $|\theta|$ increases.\textcolor{red}{{} }

Consider that $\theta'\geq0$. For every $0\leq\rho\leq L$, $\theta\geq\theta'$,
and $j\geq k$, we must have 
\[
F_{k}(\rho,\theta')\leq F_{j}(\rho,\theta')\leq F_{j}(\rho,\theta).
\]
As a result, $\min_{\rho\geq0}F_{k}(\rho,\theta')\leq\min_{\rho\ge0}F_{j}(\rho,\theta)$
for every $\theta>\theta'$. Hence the result (\ref{eq:local-optimality-condition-1})
is confirmed. 

The case of (\ref{eq:local-optimality-condition-2}) can be shown
in a similar way. 
\end{proof}

Theorem \ref{thm:convergence-two-segment} can be equivalently rewritten,
in a more general setting, using the notions and conditions from the
polar domain as follows:

\begin{thm1a}

Suppose that the set of segments $\{\mathcal{P}_{k}\}$ defined along
with $\mathscr{P}'$ satisfy the nested condition (\ref{eq:nested-segmented-model-nested-condition}).
In addition, assume that the globally optimal solution to $\mathscr{P}'$
belongs to the bounded search region $\mathcal{P}$ defined in (\ref{eq:search-region}),
and the local minimizer $\rho_{k}^{*}(\theta)$ of $F_{k}(\rho,\theta)$
is unique for each $k$ and each $\theta$, $|\theta|<\frac{\pi}{2}$.
Moreover, suppose that the functions $F_{k}(\rho,\theta)$ satisfy
conditions (\ref{eq:local-optimality-condition-0})\textendash (\ref{eq:local-optimality-condition-2}).
Then, the polar domain trajectory $(\hat{\rho}(t),\hat{\theta}(t))$
obtained from $\hat{\mathbf{x}}(t)$ following Algorithm \ref{alg:uav-positioning}
converges to the globally optimal solution to $\mathscr{P}'$.

\end{thm1a}

Note that it has been proven in Propositions \ref{prop:Search-region}
and \ref{prop:Partial-Optimality} and Lemma \ref{lem:partial-optimality}
that the objective function $f$ (which satisfies either condition
I or II) in the UAV positioning problem $\mathscr{P}$ (and, correspondingly,
the functions $F_{k}$ in $\mathscr{P}'$) satisfy the conditions
in Theorem 1A. The remaining part of this section thus focuses on
proving Theorem 1A. 

\subsection{Optimality for the Two Segment Case}

When the algorithm terminates at $t=T$, it turns out that the cost
value track record $F_{\min}(T)$ along with the algorithm trajectory
satisfies $F_{\min}(T)\leq\text{minimize}_{\rho\geq0,(\rho,\theta)\in\mathcal{P}_{2}}F_{2}(\rho,\theta)$
because 
\begin{align}
F_{\min}(T) & \leq\underset{0\leq\rho\leq L}{\text{minimize}}\quad F(\rho,0)\label{eq:thm-convergence-K2-proof-eq0}\\
 & =\underset{\rho\geq0}{\text{minimize}}\quad F(\rho,0)\label{eq:thm-convergence-K2-proof-eq1}\\
 & \leq\underset{\rho\geq0}{\text{minimize}}\quad F_{2}(\rho,0)\label{eq:thm-convergence-K2-proof-eq1b}\\
 & \leq\underset{\rho\geq0,(\rho,\theta)\in\mathcal{P}_{2}}{\text{minimize}}\quad F_{2}(\rho,\theta)\label{eq:thm-convergence-K2-proof-eq2}
\end{align}
where inequality (\ref{eq:thm-convergence-K2-proof-eq0}) is due to
Step \ref{enu:alg-initialization}) and \ref{enu:alg-final-step})
in Algorithm \ref{alg:uav-positioning} and the fact that $F_{\min}(t)$
is a non-increasing process. Equality (\ref{eq:thm-convergence-K2-proof-eq1})
is due to the condition in Theorem 1A (proven in Proposition \ref{prop:Search-region}),
inequalities (\ref{eq:thm-convergence-K2-proof-eq1b}) and (\ref{eq:thm-convergence-K2-proof-eq2})
are from conditions (\ref{eq:local-optimality-condition-1})\textendash (\ref{eq:local-optimality-condition-2})
(proven in Lemma \ref{lem:partial-optimality}).

In fact, it also holds that 
\[
F_{\min}(T)\leq\text{minimize}_{\rho\geq0,(\rho,\theta)\in\mathcal{P}_{1}}F_{1}(\rho,\theta)
\]
as follows.
\begin{lem}
[Optimality for $K=2$]\label{lem:Local-optimality} In the two
segment case, the continuous trajectory $(\rho(t),\theta(t))$ passes
through $(\rho^{*},\theta^{*})$ before it completes Step \ref{enu:alg-two-phase-update-loop})
in Algorithm \ref{alg:uav-positioning} at $t=T_{1}$, i.e., there
exists $t^{*}\leq T_{1}$, such that $\rho(t^{*})=\rho^{*}$ and $\theta(t^{*})=\theta^{*}$,
where $(\rho^{*},\theta^{*})$ is the optimal solution to the following
problem 
\begin{equation}
\underset{\rho\geq0,(\rho,\theta)\in\mathcal{P}_{1}}{\text{minmize}}\quad F_{1}(\rho,\theta)\label{eq:subproblem-P1}
\end{equation}
where the minimum value equals to $F_{\min}(T_{1})$.
\end{lem}

Since $T_{1}\leq T$, we must have $F_{\min}(T)\leq F_{\min}(T_{1})$
and the former one lower-bounds both the subproblems (\ref{eq:thm-convergence-K2-proof-eq2})
and (\ref{eq:subproblem-P1}). Since $F_{\min}(t)$ represents the
cost value track record along the trajectory $(\hat{\rho}(t),\hat{\theta}(t))$,
we conclude that $(\hat{\rho}(T),\hat{\theta}(T))$ attains the globally
optimal solution to $\mathscr{P}'$ for the two segment case. 

In the following two subsections, we first prove some preliminary
properties of Algorithm \ref{alg:uav-positioning}, and then, we use
these properties to prove Lemma \ref{lem:Local-optimality}.

\subsection{Preliminary Properties of Algorithm \ref{alg:uav-positioning}}

The following property states that, along the algorithm trajectory,
the same cost $F_{1}(\rho(t),\theta(t))$ is achievable at some prior
time $\tau\leq t$, even when the current step $(\rho(t),\theta(t))$
is not in the LOS region.
\begin{lem}
\label{lem:achievability}For every point on the algorithm trajectory
$(\rho(t),\theta(t))$, $0\leq t\leq T$, there exists $0\leq\tau\leq t$,
such that $(\rho(\tau),\theta(\tau))\in\mathcal{P}_{1}$ and 
\[
F(\rho(\tau),\theta(\tau))=F_{1}(\rho(\tau),\theta(\tau))=F_{1}(\rho(t),\theta(t)).
\]
\end{lem}
\begin{proof}
If $(\rho(t),\theta(t))\in\mathcal{P}_{1}$, we have $F(\rho(t),\theta(t))=F_{1}(\rho(t),\theta(t))$
and $\tau=t$. If $(\rho(t),\theta(t))\in\mathcal{P}_{2}$, then the
algorithm is in the loop of Step \ref{enu:alg-two-phase-update-loop}b),
where it follows the trajectory on the contour $F_{1}(\rho(t),\theta(t))=C$.
To trace backward, there must be a $0\leq\tau<t$, such that $F_{1}(\rho(\tau),\theta(\tau))=F_{1}(\rho(t),\theta(t))$
and $(\rho(\tau),\theta(\tau))\in\mathcal{P}_{1}$. Note that the
initial point $(\rho_{1}^{0},0)$ from Step \ref{enu:alg-initialization})
is in the LOS region $\mathcal{P}_{1}$ from the definition (\ref{eq:polar-representation-xy-plane}).
\end{proof}

The following result shows that the algorithm trajectory always satisfies
$\partial F_{1}(\rho,\theta(t))/\partial\rho\leq0$ until the algorithm
terminates. 
\begin{lem}
\label{lem:property-one-side}Let $\rho_{1}^{*}(\theta)$ minimize
$F_{1}(\rho,\theta)$ over all $\rho\geq0$ with $\theta$ fixed.
Then, the algorithm trajectory $(\rho(t),\theta(t))$ satisfies $\rho(t)\leq\rho_{1}^{*}(\theta(t))$
and $\partial F_{1}(\rho(t),\theta(t))/\partial\rho\leq0$ for all
$t$ before the iteration completes Step \ref{enu:alg-two-phase-update-loop}).
Moreover, $\partial F_{1}(\rho,\theta(t))/\partial\rho\leq0$ for
all $0\leq\rho\leq\rho(t)$. The same result holds for the search
trajectory in Step \ref{enu:alg-two-phase-upadte-loop-left}).
\end{lem}
\begin{proof}
First, from the definition of $\rho_{1}^{0}$ in (\ref{eq:problem-critical-point-K})
and the nested segmented propagation property (\ref{eq:nested-segmented-model-order-condition})
\textendash{} (\ref{eq:nested-segmented-model-nested-condition}),
we have $\rho_{1}^{0}\leq\rho_{1}^{*}(0)$. This is because, we have
$\mathbf{x}(\rho_{1}^{0},0)\in\mathcal{D}_{1}$ from (\ref{eq:problem-critical-point-K}),
and hence, all the points $(\rho,0)$, $0\leq\rho\leq\rho_{1}^{0}$,
are in the LOS region. As a result, if $\rho_{1}^{0}>\rho_{1}^{*}(0)$,
then $(\rho_{1}^{*}(0),0)$ is also in the LOS region (satisfying
the constraint (\ref{eq:problem-critical-point-K-constraint})), which
implies that $\rho_{1}^{*}(0)$ minimizes (\ref{eq:problem-critical-point-K}),
yielding a contradiction. Therefore, from Step \ref{enu:alg-initialization}),
we have the initial point $(\rho(0),\theta(0))$ satisfying $\theta(0)=0$
and $\rho(0)=\rho_{1}^{0}\leq\rho_{1}^{*}(0)=\rho_{1}^{*}(\theta(0))$. 

Second, as the local minimizer $\rho_{1}^{*}(\theta)$ is unique from
Proposition \ref{prop:Partial-Optimality}, we must have $\partial F_{1}(\rho,\theta)/\partial\rho<0$
for $\rho<\rho_{1}^{*}(\theta)$ and $\partial F_{1}(\rho,\theta)/\partial\rho>0$
for $\rho>\rho_{1}^{*}(\theta)$. (Note that there is no saddle point
either, due to the strict convexity under Condition 1 and monotonicity
of $f_{1}$ and $f_{2}$ under Condition 2.) Once $\partial F_{1}(\rho(t),\theta(t))/\partial\rho\geq0$,
Step \ref{enu:alg-two-phase-update-loop}) is completed. As a result,
it holds that $\rho(t)\leq\rho_{1}^{*}(\theta(t))$.
\end{proof}

\subsection{Proof of Lemma \ref{lem:Local-optimality}}

It suffices to prove for the subproblem 
\begin{equation}
\mathscr{P}'_{1+}:\quad F_{\min}(T_{1})\leq\underset{\rho\geq0,(\rho,\theta)\in\mathcal{P}_{1}^{+}}{\text{minimize}}\quad F_{1}(\rho,\theta)\label{eq:app-subproblem-P1-plus}
\end{equation}
which is essentially solved by the iterations in Step 2) of Algorithm
\ref{alg:uav-positioning}, where $\mathcal{P}_{1}^{+}=\{(\rho,\theta):0\leq\theta\leq\frac{\pi}{2},(\rho,\theta)\in\mathcal{P}_{1}\}$.
Indeed, the counterpart subproblem over the constraint set $\mathcal{P}_{1}^{-}=\{(\rho,\theta):-\frac{\pi}{2}\leq\theta\leq0,(\rho,\theta)\in\mathcal{P}_{1}\}$
is solved in a similar way by Step 4). If (\ref{eq:app-subproblem-P1-plus})
holds, then it must also hold that 
\[
F_{\min}(T_{1})\leq\text{minimize}_{\rho\geq0,(\rho,\theta)\in\mathcal{P}_{1}^{-}}F_{1}(\rho,\theta)
\]
which confirms the result of Lemma \ref{lem:Local-optimality}.

\textbf{Step A:} We first show that the algorithm trajectory $(\rho(t),\theta(t))$
in the loop of Step \ref{enu:alg-two-phase-update-loop}) can only
stop at $\theta(T)\geq\theta^{*}$, where $T\leq T_{1}$. Therefore,
as the algorithm trajectory from Step \ref{enu:alg-two-phase-update-loop})
is continuous, we must have $\theta(t)=\theta^{*}$ for some $t\leq T$. 
\begin{lem}
\label{lem:stopping-theta} The algorithm trajectory in Step \ref{enu:alg-two-phase-update-loop})
$(\rho(t),\theta(t))$ can only stop at $\theta(T)=\theta'\geq\theta^{*}$
where $(\rho^{*},\theta^{*})$ is the optimal solution to $\mathscr{P}'_{1+}$. 
\end{lem}
\begin{proof}
The result can be proven by contradiction. Suppose that Step \ref{enu:alg-two-phase-update-loop})
stops at $(\rho(T),\theta(T))$ where $\rho(T)=\rho'$ and $\theta(T)=\theta'<\theta^{*}$.
As Step \ref{enu:alg-two-phase-update-loop}) is completed, either
one of the stopping criteria should have been triggered. 

First, suppose that the condition $\partial F_{1}(\rho',\theta')/\partial\rho\geq0$
was triggered. From Lemma \ref{lem:achievability}, there exists $\tau\leq T$,
such that $F_{1}(\rho(\tau),\theta(\tau))=F_{1}(\rho',\theta')$ and
$(\rho(\tau),\theta(\tau))\in\mathcal{P}_{1}^{+}$. From Lemma \ref{lem:property-one-side}
and the condition in Theorem 1A that corresponds to Proposition \ref{prop:Partial-Optimality},
$(\rho',\theta')$ minimizes $F_{1}(\rho,\theta')$ over $\rho\geq0$.
As a result, 
\begin{align}
F_{1}(\rho(\tau),\theta(\tau)) & =F_{1}(\rho',\theta')\nonumber \\
 & =\underset{\rho\geq0}{\text{minimize}}\;F_{1}(\rho,\theta')\label{eq:app-pf-lem-stopping-eq1}\\
 & \leq\min_{j\geq1}\;\underset{\rho\geq0,\theta'<\theta\leq\frac{\pi}{2}}{\text{minimize}}F_{j}(\rho,\theta)\nonumber \\
 & \leq\underset{\rho\geq0,\theta'<\theta\leq\frac{\pi}{2}}{\text{minimize}}F_{1}(\rho,\theta)\label{eq:app-pf-lem-stopping-eq2}\\
 & \leq\underset{\rho\geq0,\theta'<\theta\leq\frac{\pi}{2},(\rho,\theta)\in\mathcal{P}_{1}^{+}}{\text{minimize}}F_{1}(\rho,\theta)\nonumber \\
 & =F_{1}(\rho^{*},\theta^{*})\nonumber 
\end{align}
where the first two inequalities are from the condition in Theorem
1A that corresponds to Lemma \ref{lem:partial-optimality}. The third
equality is by the hypothesis $\theta^{*}>\theta'$. However, this
violates the assumption $(\rho^{*},\theta^{*})$ being the solution
to subproblem $\mathscr{P}'_{1+}$, since $(\rho(\tau),\theta(\tau))\in\mathcal{P}_{1}^{+}$
now yields a lower cost. By contradiction, the stopping criterion
$\partial F_{1}(\rho',\theta')/\partial\rho\geq0$ is not satisfied. 

Second, suppose that the condition $\rho(T)=\rho'\geq L\cos\theta'$
is triggered. From the bounded search region condition in Theorem
1A (corresponding to Proposition \ref{prop:Search-region}) and from
the hypothesis $\theta(T_{3})=\theta'<\theta^{*}$ , we have 
\[
\rho'\geq L\cos\theta'>L\cos\theta^{*}\geq\rho^{*}.
\]
From Lemma \ref{lem:property-one-side}, $\rho_{1}^{*}(\theta')>\rho'$
and $\partial F_{1}(\rho,\theta')/\partial\rho<0$ for $\rho'\geq\rho\geq\rho^{*}$.
As a result,
\begin{align}
F_{1}(\rho',\theta')<F_{1}(\rho^{*},\theta') & <F_{1}(\rho^{*},\theta^{*})\label{eq:F1-inequality}
\end{align}
where the second inequality is from Lemma \ref{lem:partial-optimality}.

From Lemma \ref{lem:achievability}, there exists $\tau\leq T$, such
that $(\rho(\tau),\theta(\tau))\in\mathcal{P}_{1}^{+}$ and 
\[
F_{1}(\rho(\tau),\theta(\tau))=F_{1}(\rho',\theta')<F_{1}(\rho^{*},\theta^{*})
\]
which contradicts to the hypothesis $(\rho^{*},\theta^{*})$ being
the solution to subproblem $\mathscr{P}'_{1+}$. Therefore, the stopping
criterion $\rho(T_{3})=\rho'\geq L\cos\theta'$ is not satisfied either. 

To conclude, since neither of the stopping criteria is triggered,
by contradiction, the algorithm can only stop at $\theta(T)\geq\theta^{*}$.
\end{proof}

\textbf{Step B:} We then argue that when the trajectory reaches $\theta(t)=\theta^{*}$
for some $t\leq T$, it must hold that $(\rho(t),\theta(t))\in\mathcal{P}_{1}^{+}$
and $\rho(t)\leq\rho^{*}$.
\begin{lem}
\label{lem:P1-region-stopping} The algorithm trajectory satisfies
$(\rho(t_{1}),\theta(t_{1}))\in\mathcal{P}_{1}^{+}$, where $t_{1}\geq0$
satisfies $\theta(t_{1})=\theta^{*}$ and $\theta(t)\leq\theta^{*}$
for $t<t_{1}$. 
\end{lem}
\begin{proof}
We note that $\rho(t_{1})\leq\rho^{*}$. This is because if $\rho(t_{1})>\rho^{*}$,
we must have $F_{1}(\rho(t_{1}),\theta^{*})<F_{1}(\rho^{*},\theta^{*})$,
since $\partial F_{1}(\rho,\theta^{*})/\partial\rho<0$ for all $\rho(t_{1})>\rho>\rho^{*}$
as from Lemma \ref{lem:property-one-side}. As a result of Lemma \ref{lem:achievability},
there exists $\tau\leq t_{1}$, such that 
\[
F_{1}(\rho(\tau),\theta(\tau))=F_{1}(\rho(t_{1}),\theta^{*})<F_{1}(\rho^{*},\theta^{*})
\]
and $(\rho(\tau),\theta(\tau))\in\mathcal{P}_{1}^{+}$, which contradicts
to the assumption that $(\rho^{*},\theta^{*})$ is the solution to
(\ref{eq:subproblem-P1}). 

As $(\rho^{*},\theta^{*})\in\mathcal{P}_{1}^{+}$, from the nested
segmented propagation property (\ref{eq:nested-segmented-model-order-condition})\textendash (\ref{eq:nested-segmented-model-nested-condition}),
we can conclude that $(\rho(t_{1}),\theta(t_{1}))\in\mathcal{P}_{1}^{+}$.
\end{proof}

Since $(\rho(t_{1}),\theta(t_{1}))\in\mathcal{P}_{1}^{+}$, $\rho(t)$
will increase to $\rho^{*}$ following Step \ref{enu:alg-two-phase-update-loop}a)
in Algorithm \ref{enu:alg-initialization}. This completes the proof
that the algorithm trajectory passes through $(\rho(t_{2}),\theta(t_{2}))=(\rho^{*},\theta^{*})$
at time $t_{2}$, where $t_{1}\leq t_{2}\leq T_{3}$.

\textbf{Step C:} Now, the algorithm iterate $(\rho(t),\theta(t))$
is at the stage of Step \ref{enu:alg-two-phase-update-loop}a) in
Algorithm \ref{alg:uav-positioning}, with $\theta(t)=\theta^{*}$
and $\rho(t)\leq\rho^{*}$. It must reach $(\rho^{*},\theta^{*})$
following Step \ref{enu:alg-two-phase-update-loop}a), where $(\rho^{*},\theta^{*})$
solves $\mathscr{P}_{1+}^{'}$.

With these, we confirm the results of Lemma \ref{lem:Local-optimality}.


\subsection{Optimality for the $K$ Segment Case}

We now extend Lemma \ref{lem:Local-optimality} to the case of more
than two segments. 
\begin{lem}
[Optimality after $k$ Loops]\label{lem:partial-opt-K} After completing
the $k$th loop of Steps \ref{enu:alg-two-phase-update})\textendash \ref{enu:alg-two-phase-upadte-loop-left})
in Algorithm \ref{alg:uav-positioning} at time $t=T_{k}$, the following
holds, 
\begin{equation}
F_{\min}(T_{k})\leq\underset{\theta\geq0,(\rho,\theta)\in\mathcal{P}_{k}}{\text{minimize}}\quad F_{k}(\rho,\theta)\label{eq:per-segment-minimum}
\end{equation}
for all $k=1,2,\dots,K-1$. Moreover, (\ref{eq:per-segment-minimum})
holds for $k=K$ when the entire algorithm terminates.
\end{lem}
\begin{proof}
For the $k$th loop $(k\leq K-1)$ of Steps \ref{enu:alg-two-phase-update})
\textendash{} \ref{enu:alg-two-phase-upadte-loop-left}) in Algorithm
\ref{alg:uav-positioning}, the iteration is equivalent to that in
the two segment case, $K=2$. Specifically, the virtual propagation
segment partition $(\tilde{\mathcal{P}}_{k},\tilde{\mathcal{P}}_{k}^{\text{c}})$
in the $k$th outer loop corresponds to the LOS-NLOS partition $(\mathcal{P}_{1},\mathcal{P}_{2})$
in the loop of Steps \ref{enu:alg-two-phase-update}) \textendash{}
\ref{enu:alg-two-phase-upadte-loop-left}) for the $K=2$ case. Moreover,
the function $F_{k}$ in the $k$th outer loop corresponds to the
function $F_{1}$ in the $K=2$ case. As a result, using Lemma \ref{lem:Local-optimality},
Steps \ref{enu:alg-two-phase-update}) \textendash{} \ref{enu:alg-two-phase-upadte-loop-left})
in the $k$th loop equivalently solve 
\begin{equation}
\mathscr{P}_{k}':\quad\underset{\theta\geq0,(\rho,\theta)\in\widetilde{\mathcal{P}}_{k}}{\text{minimize}}\quad F_{k}(\rho,\theta)\label{eq:k-th-optimization}
\end{equation}
where $\tilde{\mathcal{P}}_{k}\triangleq\bigcup_{j=1}^{k}\mathcal{P}_{j}$
and $k\leq K-1$. 

We denote the minimum value of (\ref{eq:k-th-optimization}) as $\tilde{F}_{\min,k}$,
and the algorithm trajectory $(\rho(t),\theta(t))$ reaches the optimal
solution $(\hat{\rho}^{(k)},\hat{\theta}^{(k)})$ at time $t_{k}\leq T_{k}$.
Since the constraint set $\tilde{\mathcal{P}}_{k}$ in (\ref{eq:k-th-optimization})
contains the constraint set $\mathcal{P}_{k}$ in (\ref{eq:per-segment-minimum}),
it must hold that 
\[
\tilde{F}_{\min,k}\leq\mbox{minimize}_{\theta\geq0,(\rho,\theta)\in\mathcal{P}_{k}}F_{k}(\rho,\theta).
\]
Therefore, we still need to show $F_{\min}(T_{k})\leq\tilde{F}_{\min,k}$.

With such a goal, the following two cases are examined:

(i) If $(\hat{\rho}^{(k)},\hat{\theta}^{(k)})\in\mathcal{P}_{k}$,
then $F_{\min}(T_{k})\leq F(\rho(t_{k}),\theta(t_{k}))=F_{k}(\rho(t_{k}),\theta(t_{k}))=\tilde{F}_{\min,k}$
due to the track record update $F_{\min}(t)$ from Step 3a). 

(ii) If $(\hat{\rho}^{(k)},\hat{\theta}^{(k)})\notin\mathcal{P}_{k}$,
then it holds that $(\hat{\rho}^{(k)},\hat{\theta}^{(k)})\in\widetilde{\mathcal{P}}_{k}\backslash\mathcal{P}_{k}=\bigcup_{i<k}\mathcal{P}_{i}$.
\Ac{wlog}, assume that $(\hat{\rho}^{(k)},\hat{\theta}^{(k)})\in\mathcal{P}_{j}$
for some $j\leq k-1$. We thus have 
\begin{align}
\tilde{F}_{\min,k} & =F_{k}(\hat{\rho}^{(k)},\hat{\theta}^{(k)})\nonumber \\
 & \geq F_{j}(\hat{\rho}^{(k)},\hat{\theta}^{(k)})\label{eq:partial-opt-K-ieq1}\\
 & =F(\rho(t_{k}),\theta(t_{k}))\nonumber \\
 & \geq F_{\min}(t_{k})\nonumber \\
 & \geq F_{\min}(T_{k})\label{eq:partial-opt-K-ieq2}
\end{align}
where the inequality on the first line is from conditions (\ref{eq:local-optimality-condition-1})\textendash (\ref{eq:local-optimality-condition-2}).
The second line is due to the track record update $F_{\min}(t)$ from
Step 3a). This shows that inequality (\ref{eq:per-segment-minimum})
is also true for $k$. 

As a result, we have shown 
\[
F_{\min}(T_{k})\leq\tilde{F}_{\min,k}\leq\mbox{minimize}_{\theta\geq0,(\rho,\theta)\in\mathcal{P}_{k}}F_{k}(\rho,\theta)
\]
for $k=1,2,\dots,K-1$.

Finally, the last step in Algorithm \ref{alg:uav-positioning} yields
\begin{align}
F_{\min} & \leq F_{K}(\rho_{K}^{*}(0),0)\nonumber \\
 & \leq\underset{\rho\geq0,0<\theta\leq\frac{\pi}{2}}{\text{minimize}}\,F_{K}(\rho,\theta)\label{eq:partial-opt-K-ieq5}\\
 & \leq\underset{\theta\geq0,(\rho,\theta)\in\mathcal{P}_{K}}{\text{minimize}}\,F_{K}(\rho,\theta)\nonumber 
\end{align}
where the second inequality in (\ref{eq:partial-opt-K-ieq5}) is from
conditions (\ref{eq:local-optimality-condition-1})\textendash (\ref{eq:local-optimality-condition-2}).
The result of Lemma \ref{lem:partial-opt-K} is thus confirmed. 
\end{proof}

Theorem 1A is a direct result from Lemma \ref{lem:partial-opt-K},
since (\ref{eq:per-segment-minimum}) holds for $k=1,2,\dots,K$ and
$F_{\min}(t)$ is non-increasing, which implies that $F_{\min}(T_{K})$
is the global minimum value of $\mathscr{P}'$ and $(\hat{\rho}(T_{K}),\hat{\theta}(T_{K}))$
is the globally optimal solution to $\mathscr{P}'$. 

\section{Proof of Theorem \ref{thm:Maximum-Trajectory-Length}}

\label{app:pf-prop-maximum-traj-length}

An important property of Algorithm \ref{alg:uav-positioning} is that
the segment of search trajectory in Step \ref{enu:alg-two-phase-update-loop})
does not ``turn back'', and so does that in Step \ref{enu:alg-two-phase-upadte-loop-left}).
\begin{lem}
[Monotonicity]\label{lem:monotonicity} Step \ref{enu:alg-two-phase-update-loop})
in Algorithm \ref{alg:uav-positioning} strictly and monotonically
increases $\rho(t)$, and it also monotonically increases $|\theta(t)|$.
Similar property holds in Step \ref{enu:alg-two-phase-upadte-loop-left}).
\end{lem}
\begin{proof}
In Step \ref{enu:alg-two-phase-update-loop}a), $\rho(t)$ strictly
and monotonically increases, while $\theta(t)$ keeps unchanged. In
Step \ref{enu:alg-two-phase-update-loop}b), we have (omitting the
higher order term)
\begin{align*}
\rho+\Delta\rho & =\|\mathbf{x}+\Delta\mathbf{x}-\mathbf{x}_{\text{u}}\|\\
 & =\sqrt{\|\mathbf{x}-\mathbf{x}_{\text{u }}\|^{2}+2(\mathbf{x}-\mathbf{x}_{\text{u}})^{\text{T}}\Delta\mathbf{x}+\|\Delta\mathbf{x}\|^{2}}\\
 & =\rho\Big(1+\frac{1}{\rho^{2}}(\mathbf{x}-\mathbf{x}_{\text{u}})^{\text{T}}\Delta\mathbf{x}+o(\|\mathbf{x}\|)\Big)
\end{align*}
and since $(\mathbf{x}-\mathbf{x}_{\text{u}})/\rho=\mathbf{M}(\theta)\mathbf{u}$
from (\ref{eq:polar-representation-xy-plane}), we have 
\begin{align*}
\Delta\rho & =\mathbf{u}^{\text{T}}\mathbf{M}(\theta)^{\text{T}}\Delta\mathbf{x}\\
 & =\mathbf{u}^{\text{T}}\mathbf{M}(\theta)^{\text{T}}\mathbf{M}(\theta)\mathbf{u}\gamma\\
 & \qquad\qquad+\gamma\rho\mathbf{u}^{\text{T}}\mathbf{M}(\theta)^{\text{T}}\frac{d}{d\theta}\mathbf{M}(\theta)\mathbf{u}\Big(-\frac{\partial F_{k}}{\partial\theta}\Big)^{-1}\frac{\partial F_{k}}{\partial\rho}
\end{align*}
which equals to $\gamma$, being strictly positive. Therefore, Step
\ref{enu:alg-two-phase-update-loop}) strictly increases $\rho(t)$. 

In addition, Step \ref{enu:alg-two-phase-update-loop}b) moves on
the contour of $F_{k}(\rho,\theta)=C$, whose dynamics is give by
\[
\frac{\partial F_{k}(\rho,\theta)}{\partial\rho}d\rho+\frac{\partial F_{k}(\rho,\theta)}{\partial\theta}d\theta=0
\]
in which $\partial F_{k}(\rho(t),\theta(t))/\partial\rho\leq0$ according
to Lemma \ref{lem:property-one-side} (with a straight-forward generalization
from $F_{1}$ to $F_{k}$) and $\partial F_{k}(\rho,\theta)/\partial|\theta|>0$
according to Lemma \ref{lem:partial-optimality}. As a result, Step
\ref{enu:alg-two-phase-update-loop}b) monotonically increases $|\theta(t)|$. 
\end{proof}

\begin{figure}[h]
\begin{centering}
\includegraphics[width=0.7\columnwidth]{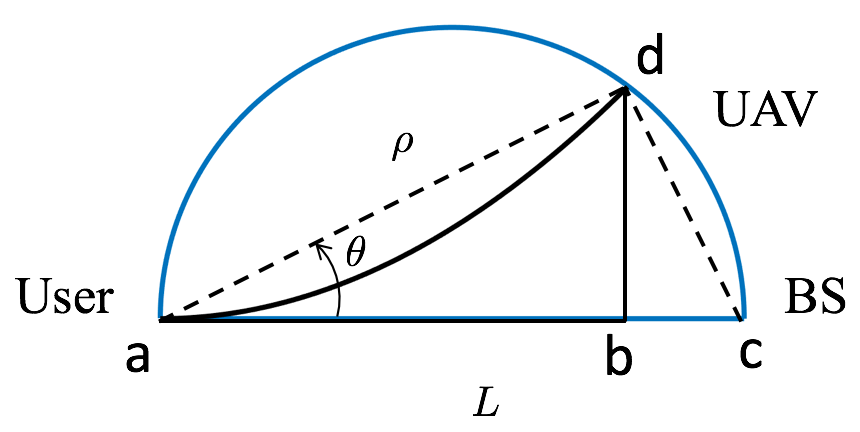}
\par\end{centering}
\caption{\label{fig:geometry-search-region} Search region and search trajectory,
where point \textbf{a} denotes user, point \textbf{c} denotes the
BS, and point \textbf{d} denotes UAV at the boundary of the search
area.}
\end{figure}
We now show that the boundary of the search region $\mathcal{P}$
in (\ref{eq:search-region}) for the $\theta>0$ branch is given by
a semi-circle as illustrated in Fig. \ref{fig:geometry-search-region}
(the blue semi-circle). To see this, we first note that the boundary
of $\mathcal{P}$ for $\theta>0$ is given by $(\rho\cos\theta,\rho\sin\theta)$,
where $\rho=L\cos\theta$ and $0\leq\theta\leq\pi/2$. Therefore,
$\overline{\text{\textbf{ad}}}=L\cos\theta$, $\overline{\text{\textbf{ac}}}=L$,
$\overline{\text{\textbf{dc}}}=\sqrt{(L\cos^{2}\theta-L)^{2}+(L\cos\theta\sin\theta-0)^{2}}$,
which yields $(\overline{\text{\textbf{ad}}})^{2}+(\overline{\text{\textbf{dc}}})^{2}=(\overline{\text{\textbf{ac}}})^{2}$
and hence point $\mathbf{d}$ is on the semi-circle with the diameter
given by line segment $\overline{\text{\textbf{ac}}}$.

From the monotone property in Lemma \ref{lem:monotonicity}, the search
trajectory is a convex curve starting from point \textbf{a} at the
user and ending at the semi-circle, \emph{e.g.}, point \textbf{d}.
Note that the length of any curve $\breve{\text{\textbf{ad}}}$ is
less than $\overline{\text{\textbf{ab}}}+\overline{\text{\textbf{bd}}}=L\cos^{2}\theta+L\cos\theta\sin\theta$,
where the maximum value over $0\leq\theta<\pi/2$ can be numerically
evaluated to be roughly $1.2L$.

Algorithm \ref{alg:uav-positioning} consists of the search on the
BS-user axis with maximum length $L$, and $K-1$ loops for the off-BS-user
axis searches, where each loop consists of searches on the left and
right branches each with maximum length $1.2L$. As a result, the
total length of the search trajectory is upper bounded by $(2.4K-1.4)L$.

\bibliographystyle{IEEEtran}

\begin{thebibliography}{10}
\providecommand{\url}[1]{#1}
\csname url@samestyle\endcsname
\providecommand{\newblock}{\relax}
\providecommand{\bibinfo}[2]{#2}
\providecommand{\BIBentrySTDinterwordspacing}{\spaceskip=0pt\relax}
\providecommand{\BIBentryALTinterwordstretchfactor}{4}
\providecommand{\BIBentryALTinterwordspacing}{\spaceskip=\fontdimen2\font plus
\BIBentryALTinterwordstretchfactor\fontdimen3\font minus
  \fontdimen4\font\relax}
\providecommand{\BIBforeignlanguage}[2]{{%
\expandafter\ifx\csname l@#1\endcsname\relax
\typeout{** WARNING: IEEEtran.bst: No hyphenation pattern has been}%
\typeout{** loaded for the language `#1'. Using the pattern for}%
\typeout{** the default language instead.}%
\else
\language=\csname l@#1\endcsname
\fi
#2}}
\providecommand{\BIBdecl}{\relax}
\BIBdecl

\bibitem{ZenZhaLim:J16}
Y.~Zeng, R.~Zhang, and T.~J. Lim, ``Wireless communications with unmanned
  aerial vehicles: opportunities and challenges,'' \emph{{IEEE} Commun. Mag.},
  vol.~54, no.~5, pp. 36--42, 2016.

\bibitem{MozSaaBenDeb:J16c}
M.~Mozaffari, W.~Saad, M.~Bennis, and M.~Debbah, ``Unmanned aerial vehicle with
  underlaid device-to-device communications: Performance and tradeoffs,''
  \emph{{IEEE} Trans. Wireless Commun.}, vol.~15, no.~6, pp. 3949--3963, 2016.

\bibitem{XiaXiaXia:M16}
Z.~Xiao, P.~Xia, and X.-G. Xia, ``Enabling {UAV} cellular with millimeter-wave
  communication: Potentials and approaches,'' \emph{{IEEE} Commun. Mag.},
  vol.~54, no.~5, pp. 66--73, 2016.

\bibitem{VanChiPol:J16}
B.~Van Der~Bergh, A.~Chiumento, and S.~Pollin, ``{LTE} in the sky: trading off
  propagation benefits with interference costs for aerial nodes,'' \emph{{IEEE}
  Commun. Mag.}, vol.~54, no.~5, pp. 44--50, 2016.

\bibitem{MerGuv:C15}
A.~Merwaday and I.~Guvenc, ``{UAV} assisted heterogeneous networks for public
  safety communications,'' in \emph{Proc. Wireless Commun. and Networking Conf.
  Workshops}, 2015, pp. 329--334.

\bibitem{HoSujJohDe:C13}
D.-T. Ho, P.~Sujit, T.~A. Johansen, and J.~B. De~Sousa, ``Performance
  evaluation of cooperative relay and particle swarm optimization path planning
  for {UAV} and wireless sensor network,'' in \emph{Proc. IEEE Globecom
  Workshops}, 2013, pp. 1403--1408.

\bibitem{sharma2016uav}
V.~Sharma, M.~Bennis, and R.~Kumar, ``{UAV}-assisted heterogeneous networks for
  capacity enhancement,'' \emph{{IEEE} Commun. Lett.}, vol.~20, no.~6, pp.
  1207--1210, 2016.

\bibitem{OuyZhuLinLiu:J14}
O.~Jian, Z.~Yi, L.~Min, and L.~Jia, ``Optimization of beamforming and path
  planning for {UAV}-assisted wireless relay networks,'' \emph{Chinese Journal
  of Aeronautics}, vol.~27, no.~2, pp. 313--320, 2014.

\bibitem{ZenZha:J17}
Y.~Zeng and R.~Zhang, ``Energy-efficient {UAV} communication with trajectory
  optimization,'' \emph{{IEEE} Trans. Wireless Commun.}, vol.~16, no.~6, pp.
  3747--3760, 2017.

\bibitem{JinZhaCha:C12}
Y.~Jin, Y.~D. Zhang, and B.~K. Chalise, ``Joint optimization of relay position
  and power allocation in cooperative broadcast wireless networks,'' in
  \emph{Proc. IEEE Int. Conf. Acoustics, Speech, and Signal Processing}, 2012,
  pp. 2493--2496.

\bibitem{ChoJunSun:C14}
D.~H. Choi, B.~H. Jung, and D.~K. Sung, ``Low-complexity maneuvering control of
  a {UAV}-based relay without location information of mobile ground nodes,'' in
  \emph{Proc. IEEE Symposium on Computers and Commun.}, 2014, pp. 1--6.

\bibitem{ChaChaGagGag:C14}
A.~Chamseddine, G.~Charland-Arcand, O.~Akhrif, S.~Gagn{\'e}, F.~Gagnon, and
  D.~Couillard, ``Optimal position seeking for unmanned aerial vehicle
  communication relay using only signal strength and angle of arrival,'' in
  \emph{IEEE Conf. on Decision and Control}, Dec 2014, pp. 976--981.

\bibitem{JiaSwi:J12}
F.~Jiang and A.~L. Swindlehurst, ``Optimization of {UAV} heading for the
  ground-to-air uplink,'' \emph{{IEEE} J. Sel. Areas Commun.}, vol.~30, no.~5,
  pp. 993--1005, 2012.

\bibitem{AzaRosChePol:J18}
M.~M. Azari, F.~Rosas, K.-C. Chen, and S.~Pollin, ``Ultra reliable {UAV}
  communication using altitude and cooperation diversity,'' \emph{{IEEE} Trans.
  Commun.}, vol.~66, no.~1, pp. 330--344, 2018.

\bibitem{ur2018uav}
S.~ur~Rahman and Y.-Z. Cho, ``{UAV} positioning for throughput maximization,''
  \emph{EURASIP J. on Wireless Commun. and Networking}, vol. 2018, no.~1,
  p.~31, 2018.

\bibitem{lyu2017placement}
J.~Lyu, Y.~Zeng, R.~Zhang, and T.~J. Lim, ``Placement optimization of
  {UAV}-mounted mobile base stations,'' \emph{{IEEE} Commun. Lett.}, vol.~21,
  no.~3, pp. 604--607, 2017.

\bibitem{AlhKanJam:C14}
A.~Al-Hourani, S.~Kandeepan, and A.~Jamalipour, ``Modeling air-to-ground path
  loss for low altitude platforms in urban environments,'' in \emph{Proc. IEEE
  Global Telecomm. Conf.}, 2014, pp. 2898--2904.

\bibitem{MozSaaBenDeb:C15}
M.~Mozaffari, W.~Saad, M.~Bennis, and M.~Debbah, ``Drone small cells in the
  clouds: Design, deployment and performance analysis,'' in \emph{Proc. IEEE
  Global Telecomm. Conf.}, 2015, pp. 1--6.

\bibitem{HouSitLar:J14}
A.~Hourani, K.~Sithamparanathan, and S.~Lardner, ``Optimal {LAP} altitude for
  maximum coverage,'' \emph{{IEEE} Commun. Lett.}, no.~99, pp. 1--4, 2014.

\bibitem{MozSaaBenDeb:C16}
M.~Mozaffari, W.~Saad, M.~Bennis, and M.~Debbah, ``Optimal transport theory for
  power-efficient deployment of unmanned aerial vehicles,'' in \emph{Proc. IEEE
  Int. Conf. Commun.}, Kuala Lumpur, Malaysia, May 2016, pp. 1--6.

\bibitem{MozSaaBenDeb:J16b}
------, ``Efficient deployment of multiple unmanned aerial vehicles for optimal
  wireless coverage,'' \emph{{IEEE} Commun. Lett.}, vol.~20, no.~8, 2016.

\bibitem{FenGeeTamNix:C06}
Q.~Feng, J.~McGeehan, E.~K. Tameh, and A.~R. Nix, ``Path loss models for
  air-to-ground radio channels in urban environments,'' in \emph{Proc. IEEE
  Semiannual Veh. Technol. Conf.}, vol.~6, 2006, pp. 2901--2905.

\bibitem{TR36777}
\BIBentryALTinterwordspacing
``Study on enhanced {LTE} support for aerial vehicles (release 15),'' 3GPP TR
  36.777, Tech. Rep., 2017. [Online]. Available: \url{http://www.3gpp.org}
\BIBentrySTDinterwordspacing

\bibitem{CheGes:C17}
J.~Chen and D.~Gesbert, ``Optimal positioning of flying relays for wireless
  networks: A {LOS} map approach,'' in \emph{Proc. IEEE Int. Conf. Commun.},
  Paris, France, May 2017.

\bibitem{BaiHea:J15}
T.~Bai and R.~W. Heath, ``Coverage and rate analysis for millimeter-wave
  cellular networks,'' \emph{{IEEE} Trans. Wireless Commun.}, vol.~14, no.~2,
  pp. 1100--1114, 2015.

\bibitem{CheYanGes:C16}
J.~Chen, U.~Yatnalli, and D.~Gesbert, ``Learning radio maps for {UAV}-aided
  wireless networks: A segmented regression approach,'' in \emph{Proc. IEEE
  Int. Conf. Commun.}, Paris, France, May 2017.

\bibitem{CheEsrGesMit:C17}
J.~Chen, O.~Esrafilian, D.~Gesbert, and U.~Mitra, ``Efficient algorithms for
  air-to-ground channel reconstruction in {UAV}-aided communications,'' in
  \emph{Proc. IEEE Global Telecomm. Conf.}, Dec. 2017, {Wi}-{UAV} workshop.

\bibitem{NabBolKne:J04}
R.~U. Nabar, H.~Bolcskei, and F.~W. Kneubuhler, ``Fading relay channels:
  Performance limits and space-time signal design,'' \emph{{IEEE} J. Sel. Areas
  Commun.}, vol.~22, no.~6, pp. 1099--1109, 2004.

\bibitem{LanTseWor:J04}
J.~N. Laneman, D.~N. Tse, and G.~W. Wornell, ``Cooperative diversity in
  wireless networks: Efficient protocols and outage behavior,'' \emph{{IEEE}
  Trans. Inf. Theory}, vol.~50, no.~12, pp. 3062--3080, 2004.

\bibitem{WanCanGiaLan:J07}
T.~Wang, A.~Cano, G.~B. Giannakis, and J.~N. Laneman, ``High-performance
  cooperative demodulation with decode-and-forward relays,'' \emph{{IEEE}
  Trans. Commun.}, vol.~55, no.~7, pp. 1427--1438, 2007.

\end{thebibliography}


\end{document}